\documentclass[aps,pra,onecolumn,superscriptaddress,longbibliography,showkeys]{revtex4-2}

\usepackage{multirow}
\usepackage{wrapfig}
\usepackage{geometry}
\usepackage{ragged2e}
\usepackage{enumitem}
\usepackage{amsmath}
\usepackage{graphicx}
\usepackage{amsthm}
\usepackage{amsmath}
\usepackage{amsfonts}
\usepackage{color}
\usepackage{epstopdf}
\usepackage{epsfig}

\newtheorem{Xeorem}{Theorem}[section]
\newtheorem{definition}[Xeorem]{Definition}

\newtheorem{resultttt}[Xeorem]{\bf Quantum Indeterminacy Postulate}

\newtheorem{lemma}[Xeorem]{Lemma}

\usepackage{float}

\begin{document}

\def\crta{\vrule height1.41ex depth-1.27ex width0.34em}
\def\dj{d\kern-0.36em\crta}
\def\Crta{\vrule height1ex depth-0.86ex width0.4em}
\def\Dj{D\kern-0.73em\Crta\kern0.33em}
\dimen0=\hsize \dimen1=\hsize \advance\dimen1 by 40pt

\title{Quantum Contextual Hypergraphs, Operators, Inequalities, and Applications in Higher Dimensions}

\author{Mladen Pavi\v ci\'c}
\email{mpavicic@irb.hr}
\homepage{https://www2.irb.hr/users/mpavicic;
  https://mladen-pavicic.github.io/web/; Orcid: 0000-0003-1915-6702}

\affiliation{{Center} of Excellence for Advanced
  Materials and Sensors, Institute Ru{\dj}er
  Bo\v skovi\'c, 10000 Zagreb, Croatia}
\affiliation{Institute of Physics, 10000 Zagreb, Croatia}
\affiliation{Nano Optics, Department of Physics, Humboldt University,
12489 Berlin, Germany}

\begin{abstract}
  Quantum contextuality plays a significant role in
  supporting quantum computation and quantum information theory.
  The key tools for this are the Kochen--Specker and
  non-Kochen--Specker contextual sets. Traditionally, their
  representation has been predominantly operator-based, mainly
  focusing on specific constructs in dimensions ranging from three
  to eight. However, nearly all of these constructs can be represented
  as low-dimensional hypergraphs. This study demonstrates how to
  generate contextual hypergraphs in any dimension using various
  methods, particularly those that do not scale in complexity
  with increasing dimensions. Furthermore, we introduce innovative
  examples of hypergraphs extending to dimension 32. Our methodology
  reveals the intricate structural properties of hypergraphs,
  enabling precise quantifications of contextuality of implemented
  sets. Additionally, we investigate several promising applications
  of hypergraphs in quantum communication and quantum computation,
  paving the way for future breakthroughs in the field.
\end{abstract}


\maketitle 

\section{Introduction}\label{sec:intro}

Quantum contextuality is a property of a set of quantum
states that precludes assignments of predetermined
0--1 values to them, i.e.,~it assumes the total absence of
two-valued states on them. The~set can be a Kochen--Specker
one~\cite{zimba-penrose,matsuno-07,pavicic-quantum-23}
or a non-Kochen--Specker one~\cite{pavicic-entropy-23}.
We will provide relevant definitions and references later~on.

Contextual sets have been implemented in various experiments
~\cite{simon-zeil00,michler-zeil-00,amselem-cabello-09,liu-09,d-ambrosio-cabello-13,ks-exp-03,lapkiewicz-11,zu-wang-duan-12,canas-cabello-8d-14,canas-cabello-14,zhan-sanders-16,li-zeng-17,li-zeng-zhang-19,frustaglia-cabello-16,zhang-zhang-19,h-rauch06,cabello-fillip-rauch-08,b-rauch-09,k-cabello-blatt-09,moussa-09,jerger-16},
which we will not discuss here, because most of them simply confirm
that implemented contextual sets are~contextual. 

Quantum contextual sets have found application in quantum
communication~\cite{cabello-dambrosio-11,nagata-05,saha-hor-19},
quantum computation~\cite{magic-14,bartlett-nature-14}, quantum
nonlocality~\cite{kurz-cabello-14}, quantum steering
~\cite{tavakoli-20}, and~lattice theory
~\cite{mp-7oa}, but mainly with
smallest sets and with little elaboration regarding how to scale them up. 
We shall discuss several application proposals and possible
scale-ups later~on.

First of all, we shall provide hypergraph formulations for
Kochen--Specker as well as for non-Kochen--Specker contextual sets.
These formulations may arise, e.g.,~from
\begin{itemize}
\item operator-based sets
~\cite{klyachko-08,yu-oh-12,beng-blan-cab-pla12,xu-chen-su-pla15,ram-hor-14,cabell-klein-budr-prl-14};  
\item sets built by multiples of mutually orthogonal vectors,
  where at least one of the multiples contains less than $n$
  vectors, with $n$ being the dimension of space in which the sets
  reside \cite{magic-14}; or 
\item the so-called true-implies-false and true-implies-true
  sets~\cite{cabello-svozil-18,svozil-21}, etc. 
\end{itemize}

To make this introduction self-reliant, we need to introduce several
definitions. We do not assume that the reader is familiar with the
language and formalism of general hypergraphs or the particular
types of hypergraphs that we employ in this paper.
Therefore, we shall briefly and informally review some properties
of general hypergraphs, followed by a more rigorous definition of
the specific types that we use.

\subsection{\label{subsec:g-hypergraphs}General
  Hypergraphs}

A general hypergraph is a pair of a finite set of elements and a
family of subsets of these elements. The~elements are called
vertices of the hypergraph, and~the subsets are called hyperedges
of the hypergraph. Vertices might be represented by vectors,
operators, numbers, geometrical points, files in a database,
elements of a DNA sequence, or~other objects, and~hyperedges are
represented by a relation between the vertices contained in them,
such as orthogonality, inclusion, geometry, data records, genes,
etc.~\cite{berge-73,berge-89,bretto-13,voloshin-09}. Thus, a~set
of {\em vertices}, $V=\{v_1,v_2,\dots,v_k\}$ and a set of subsets
of $V$ (called {\em {hyperedges}}: $e_i$, $i=1,\dots,l$),
$E=\{e_1,e_2,\dots,e_l\}$, build a pair ${\cal{H}}=(V,E)$ called
a {\em hypergraph}; notations $V({\cal{H}})$ and $E({\cal{H}})$
are also in use. Graphically, a~hyperedge $e_j$ may be represented
as a continuous curve joining two elements/points/vertices if the
cardinality (number of elements) within the hyperedge is $|e_j|=2$,
by a loop if $|e_j|=1$, and~by a closed curve enclosing the elements
if $|e_j|>2$. Numerically, they can be represented by the incidence
matrices (\cite{berge-89}, p.~2, Figure~1), with columns as
hyperedges and rows as vertices. Intersections of hyperedge columns
with vertex rows contained in hyperedges are assigned `1', and
those not contained are assigned `0'.

The number of vertices within a hypergraph ($k$), i.e.,
the cardinality of $V$ ($|V|$), is called the {\em order} of
a hypergraph, and~the number of hyperedges within a hypergraph
($l$), i.e.,~the cardinality of $E$ ($|E|$), is called the
{\em size} of a~hypergraph. 

\subsection{\label{subsec:mmps}MMP~Hypergraphs}

The particular type of general hypergraph that we deal with in this paper
is the McKay--Megill--Pavi\v ci\'c hypergraph (MMPH).

\begin{definition}\label{def:MMP-dim}{\bf MMPH-dimension} 

  $n$
  ({\em MMPH
  -dim} $n$) is a predefined (for an assumed task or
  purpose) maximal possible number ($n$) of vertices within a
  hyperedge of an {\em MMPH}, even when none of the processed
  hyperedges include $n$ vertices. It is abbreviated as
  {\em MMPH-dim}.
\end{definition}

\begin{definition}\label{def:MMP-string}
  An {\bf {MMPH}} is a connected  hypergraph
  ${\cal{H}}=(V,E)$ (where $V=\{V_1,V_2,\dots,V_k\}$ is a set of
  {\em vertices} and $E=\{E_1,E_2,\dots,E_l\}$ sets of
  {\em hyperedges}) of {\em MMPH-dim} $n\ge 3$ in which
\begin{enumerate}
  \item Every vertex belongs to at least one hyperedge;
  \item Every hyperedge contains at least $2$ and at most $n$
    vertices;
  \item No hyperedge shares only one vertex with another
    hyperedge;
  \item Hyperedges may intersect with each other in at most $n-2$
    vertices;
  \item Numerically, an {\em MMPH} 
    is a string of ASCII
    characters corresponding to vertices and organized in
    substrings (separated by commas) corresponding to hyperedges;
    each string ends with a period {\rm(}$\!$`\/{\tt .}$\!$'\/{\rm)};
    one uses 90 characters: 
    {{\tt 1$\,$\dots$\,$9 A$\,$\dots$\,$Z a$\,$\dots$\,$z$\,$!$\,$"$\,$\#}$\,${\rm{\$}$\,$\%$\,$\&$\,$'$\,$( ) * - / : ; \textless\ = \textgreater\ ?
    @ [ {$\backslash$} ] \^{} \_ {`} {\{} {\textbar} \} $\sim$}};
    when exhausted, one reuses them prefixed by `{\/\tt +}$\!$' and
    then by `\/{\tt ++}$\!$', etc.; there is no limit on their
    length;
  \item Graphically, vertices are represented as dots and
    hyperedges as (curved) lines passing through them.
  \end {enumerate}
\end{definition}

The differences between the standard hypergraph formalism and the
MMPH formalism are illustrated in (\cite{pavicic-quantum-23}, Figure~1).

\subsection{\label{subsec:b-non-b}Non-Binary and
  Binary MMPHs}
\vspace{6pt}

\begin{definition} A $k$-$l$ {\em MMPH} of dim $n\ge 3$
  with $k$ vertices and $l$ hyperedges, whose $j$-th hyperedge
  contains $\kappa(j)$ vertices  $(2\le\kappa(j)\le n$,\
  $j=1,\dots,l)$, to which it is impossible to assign {1}s and
  {0}s in such a way that the following {\em rules}
  hold
\begin{enumerate}[label=(\roman*)]
\item No two vertices within any of its hyperedges may both be
assigned the value $1$;
\item In any of its hyperedges, not all vertices may be
assigned the value $0$.
\end{enumerate}
is called a  {\bf non-binary} {\em MMPH} {\rm(}{\bf NBMMPH}{\rm)}. 

\label{def:n-b}
\end{definition}

An NBMMPH is contextual as it does not allow predetermined
values 1 and 0 to be assigned to all vertices
by violating rules ({\em i}) and ({\em ii}).

\begin{definition} An {\em MMPH} to which it is possible to
  assign 1s and 0s to satisfy rules (i)
  and (ii) of Definition~\ref{def:n-b} is called a
  {\bf {binary}} {\em MMPH} {\rm(}{\bf {BMMPH}}{\rm)}.
\label{def:bin}
\end{definition}

A BMMPH is noncontextual as it does allow predetermined values
1 and 0 to be assigned to all vertices by satisfying rules
({\em i}) and ({\em ii}).

\begin{definition} A {\bf {critical NBMMPH}} is an
  {\em NBMMPH} that is minimal in the sense that removing
  any of its hyperedges transforms it into a {\em BMMPH}.
\label{def:critical}
\end{definition}

\begin{definition}\label{df:multi} {\bf {Vertex multiplicity}} is
  the number of hyperedges that vertex `$i$' belongs to; we denote
  it by $m(i)$.
\end{definition}

\begin{definition}\label{df:class} A {\bf {master}} is a
  non-critical {\em MMPH} that contains smaller critical and
  non-critical sub-{\em MMPH}s. A~collection of all sub-{\em MMPH}s
  of an {\em MMPH} master forms its {\em class}. 
\end{definition}

\begin{definition}\label{def:over-sub}
  Let ${\cal K}$ be a subset of an {\em MMPH} ${\cal H}=(V,E)$.
  A ${\cal K}$ from which at least one of vertices with vertex
  multiplicity $m=1$ is taken out is called a 
  {$\overline{\rm \bf subhypergraph}$}.
\end{definition}

\subsection{\label{subsec:structure}Hypergraph Structural
   Discriminators}

The principal distinction between an NBMMPH and a BMMPH is that,
in the latter, we can assign 1s to vertices to satisfy conditions
{\em (i)} and {\em (ii)} from Definition~\ref{def:n-b} and that
all hyperedges contain one of them, while, in the former, this is
not possible since there exists at least one hyperedge where the
conditions are violated, preventing any vertex from being assigned
the value 1. It has been found that the number of 1s that one can
assign to vertices in an NBMMPH vs.~a BMMPH plays a crucial role
in discriminating between these two MMPHs and in evaluating the
contextuality of the~former. 
 
\begin{definition}
  The {\bf MMPH classical vertex index} $HI_c$ is the number
  of $1$s that one can assign to the vertices of an {\em MMPH} so as
  to satisfy conditions {\em(i)} and {(ii)} from
  Definition~{\em\ref{def:n-b}}. The maximal (minimal) $HI_c$ is denoted
  as $HI_{cM}$ ($HI_{cm}$).
\label{def:ch-i}
\end{definition}

(Note: In ref.~\cite{pavicic-entropy-19}, some values of
$HI_c$ were incorrectly calculated due to an application problem
in our previous algorithm and program; algorithm and program
{\textsc{One}} used since (see~\cite{pavicic-quantum-23,pw-23a})
are substitutes for the previous ones.)

\begin{definition}
  The {\bf MMPH classical multiplexed vertex index} $HI^m_c$ is
  the number that we obtain when summing up all multiplicities of
  vertices of an $n$-dim {\em MMPH}, whose every hyperedge
  contains $n$ vertices, to~which we can assign {1}s so as to
  satisfy conditions (i) and (ii) from Definition~{\ref{def:n-b}}.
  The maximal (minimal) $HI^m_c$ is denoted as
  $HI^m_{cM}$ ($HI^m_{cm}$).
\label{def:ch-mi}
\end{definition}

We obtain $HI_c$ and $HI^m_c$ through the algorithm and its program
{\textsc{One}}, which assigns 1s to the vertices of an MMPH.
The algorithm randomly searches for a distribution of 1s satisfying
conditions {\em(i)} and {\em(ii)} from Definition~\ref{def:n-b}.
It starts with a randomly chosen hyperedge in which one vertex is
assigned 1 and the other vertices are assigned 0s and continues with
the connected hyperedges until all permitted vertices are assigned 1.
Multiplicities for the found 1s accumulated in the process are taken
into account. For~NBMMPHs, this means ``until a contradiction is
reached'', i.e.,~to a point at which no vertex from the remaining
hyperedges can be assigned 1; vertices within the latter hyperedges
are all assigned 0s. The~maximal number of \mbox{1s} ($HI_{cM}$,
$HI^m_{cM}$) is obtained by (up to 50,000) parallel runs with
reshuffled vertices and hyperedges. Because~we do not make use of
backtracking in the search algorithm to resolve conflicts, the
procedure does not exponentially increase the CPU time with an
increasing number of vertices. MMPHs with several thousand vertices
and hyperedges are processed within seconds on each CPU of a cluster
or a~supercomputer.

The probability of not finding the correct minimal or maximal $HI_c$
and $HI^m_c$ after so many runs is extremely small;
nevertheless, this small probability restrains our results,
meaning that slightly larger maxima and smaller minima might
be found in future computations for a chosen~hypergraph.

\begin{definition}The {\bf classical hyperedge number $l_c$}
  is the number of hyperedges that contain vertices that take
  part in building $HI_{c}$, and the maximal and minimal number
  of such hyperedges are $l_{cM}$ and $l_{cm}$, respectively.
\label{def:lcMm}
\end{definition}

We stress that, in~most cases, $l_{cM}$ hyperedges do not contain
$HI_{cM}$ vertices but a smaller number of them. Moreover, $l_{cm}$
hyperedges usually do not contain $HI_{cm}$ vertices but a
larger number of~them.

The classical vertex index $HI_{cM}$ of a hypergraph $\cal H$ is
related to the independence number of $\cal H$ introduced by
Gr{\"o}tschel, Lov{\'a}sz, and~Schrijver (GLS)
 (\cite{gro-lovasz-schr-81}, p.~192). They introduced the
definition for graphs, but it holds for hypergraphs as well,
with graph cliques transliterated into~hyperedges. 

\begin{definition}{\bf GLS ${\alpha}$}. The~independence
  number of $\cal H$ denoted by $\alpha({\cal H})$ is the maximum
  number of pairwise non-adjacent vertices.  
\label{def:alpha}
\end{definition}

The independence number $\alpha$ has been given several definitions
and names in the literature. For~instance, ``$\alpha({\cal H})$
is the size of the largest set of vertices of $\cal H$ such that
no two elements of the set are adjacent'' \cite{magic-14}.
Such a set is called an {\em independent} or a {\em stable}
set (\cite{amaral-cunha-18}, Definition~2.13) (\cite{berge-73}, p.~272, 428),
and $\alpha$ is also called a {\em stability number}
( \cite{berge-73}, p.~272, 428). In~such a set, no two vertices are
connected by a hyperedge. The definitions of these notions given by
Voloshin differ since his sets might include two or more vertices
from the same hyperedge (\cite{voloshin-09}, p.~151). 

\begin{lemma} $HI_{cM}({\cal H})=\alpha({\cal H})$.
  \label{lemma:alpha}
\end{lemma}

\begin{proof} Via conditions (i) and (ii) from
  Definition~\ref{def:n-b} which Definition~\ref{def:ch-i} invokes,
  no two vertices to which one can assign 1 can belong to
  the same hyperedge. The~maximum number of such vertices,
  i.e.,~$H_{cM}({\cal H})$, is therefore the maximum number of
  pairwise non-adjacent vertices, i.e.,~according to
  Definition~\ref{def:alpha}, just $\alpha({\cal H})$.
\end{proof}  

\subsection{\label{subsec:coord}Coordinatization}

Vector or state or operator representation, i.e.,~a
{\em coordinatization} of vertices, is operationally required
for any implementation of an MMPH since a full coordinatization
of vertices transforms MMPH-dim $n$ into a dimension of a Hilbert
space determined by the vectors that each vertex is assigned to.
Whether we refer to an MMPH with or without a coordinatization
will be clear from the~context.

An $n$-dim $k$-$l$ NBMMPH ${\cal{H}}$ need not have
a coordinatization, but,~when it does, the~vertices in every
hyperedge have definite mutually orthogonal vectors assigned to
them. This means that each hyperedge  $E_j$, $j=1,\dots,l$ should
have not only $\kappa(j)$ vectors corresponding to its $\kappa(j)$
vertices specified, but~also $n-\kappa(j)$ ones that must exist in
each $E_j$ by the virtue of orthogonality in the $n$-dim space, so
as to form an orthogonal basis of the space. Such an extended
${\cal{H}}$ is called a {\em filled} ${\cal{H}}$. We can have a
filled MMPH even when a coordinatization does not exist. Then,
it means that, to each hyperedge that does not contain $n$ vertices,
vertices are added so that it~does. 

\begin{definition}\label{df:filled}
  A {\bf filled} $n$-dim {\em MMPH} is one that is derived
  from an $n$-dim {\em MMPH} with at least one hyperedge
  containing fewer than $n$ vertices by adding vertices to ensure
  that all hyperedges have precisely $n$ vertices.
\end{definition}
\begin{wrapfigure}{r}{0.3\textwidth}
\vspace{-20pt}
  \begin{center}
    \includegraphics[width=0.3\textwidth]{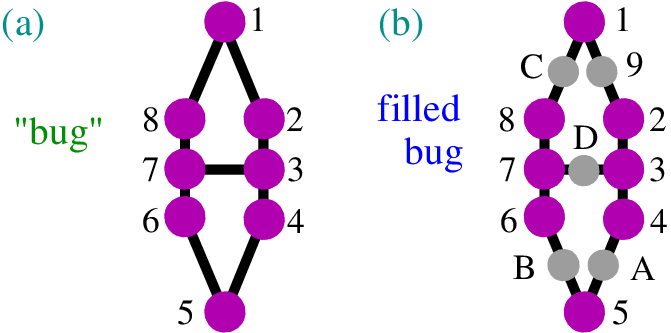}
  \end{center}
\vspace{-10pt}
  \caption{(\textbf{a})  The 8-7 NBMMPH 
    (\cite{pavicic-pra-22}, Supplementary Materials, Figure~3) or the
    {\em bug}; cf.~(\cite{koch-speck}, $\Gamma_1$); 
    \mbox{(\textbf{b}) filled} bug---13-7 BMMPH; grey dots represent
    vertices with $m=1$.}
\label{fig:bug}
\end{wrapfigure}

For instance, a~contextual critical NBMMPH
{\tt 12,234,45,56,\break 678,81,37}, or a ``bug'', as shown in
Figure~\ref{fig:bug}(a), obtains a coordinatization from its filled
MMPH {\tt 192,234,4A5,5B6,678,}\linebreak{\tt 8C1,3D7}. The latter
MMPH is not contextual; it is a BMMPH, as shown in
Figure~\ref{fig:bug}(b).

Our algorithms and programs can detect the contextuality of an
MMPH regardless of whether its coordinatization is given (or even
existent) or not. The handling of MMPHs using our algorithms
embedded in the programs SHORTD, MMPSTRIP, MMPSUBGRAPH, VECFIND,
STATES01, and others, without taking their coordinatization into
account, gives us a computational advantage over handling them
with a coordinatization, because processing bare hypergraphs is
faster than processing them with vectors assigned to their vertices.

\begin{definition}\label{df:coord}
  A {\em coordinatization} of an {\em MMPH} is a set
  of vectors/states assigned to its vertices, which is a subset
  of $n$-dim vectors in ${\cal H}^n$, $n\ge 3$, assigned to its
  vertices, provided that all hyperedges contain $n$ vertices,
  or~to vertices of its filled {\em MMPH} or of any of its
  {\em MMPH} masters whose hyperedges all contain $n$ vertices.
\end{definition}

Hence, an~NBMMPH whose hyperedges contain $m\le n$ vertices inherits
its coordinatization from the coordinatization of its master or of
its filled set (they may coincide, but~usually they do not). When
the method {\bf M1} (from Section~\ref{subsec:nonks}) for the
generation of MMPHs is applied,  a coordinatization is automatically
assigned to each contained MMPH by the very procedure of its
generation (method {\bf M2}) from master MMPHs, as we shall
see below. 

Note that the Kochen--Specker theorem below does not require a
coordinatization, although~its original proof involved one~\cite{koch-speck}.

\begin{Xeorem} {\bf Kochen--Specker Theorem.} 
  There exist $k$-$l$ {\em MMPH}s of dim $n\ge 3$ with $k$ vertices
  and $l$ hyperedges, whose $j$-th hyperedge contains $\kappa(i)$
  vertices  $(2\le\kappa(j)\le n$,\ $j=1,\dots,l)$ to which it is
  impossible to assign {1}s and {0}s in such a way that the
  following {\em rules} hold:
\begin{enumerate}[label=(\roman*)]
\item No two vertices within any of its hyperedges may both be
assigned the value $1$;
\item In any of its hyperedges, not all vertices may be
assigned the value $0$.
\end{enumerate}
Such {\em MMPH}s are called {\bf KS MMPH}s. All {\em KS MMPH}s
are {\em NBMMPH}s. Every {\em KS MMPH} is a {\bf proof} of the
{\em KS} theorem. 
\label{def:ks-theorem}
\end{Xeorem}

\begin{proof} Obvious~\cite{koch-speck,budroni-cabello-rmp-22,zimba-penrose,pavicic-quantum-23}.
\end{proof}

This paper is organized as~follows.

In Section~\ref{subsec:ks-non-ks}, we introduce Kochen--Specker (KS)
and non-Kochen--Specker (non-KS) NBMMPHs.

In Section~\ref{subsec:nonks}, we present eight methods, {\bf M1--8},
which we make use of to generate~MMPHs.

In Section~\ref{subsec:ineq}, we consider the operator approach, which
yields several types of inequalities, and we prove several lemmas on
them. We also consider the fractional independence (packing) number
$\alpha^*({\cal H})$ 

 of an {MMP} and arrive at the
Quantum Indeterminacy Postulate \ref{postulate} and two different
types of statistics, Definitions~\ref{def:raw} and \ref{def:post}, which
satisfy Theorem \ref{eq:sum-edge}.

In Section~\ref{subsec:gen}, we review special cases of KS and non-KS
MMPHs in dimensions 3 to 32 and provide new instances of them. We
also introduce a new graphical presentation of higher-dimensional
MMPHs. 

In Section~\ref{subsec:app}, we give four possible applications of
higher-dimensional MMPHs. Section~\ref{appA} considers larger alphabet
communication, and we provide a detailed protocol for it;
Section~\ref{appB} presents the oblivious communication protocol;
Section~\ref{appC} discusses generalized Hadamard matrices;
and Section~\ref{appD} discusses stabilizer~operations.

In Section~\ref{sec:disc}, we discuss the results obtained in this
article. 

\section{\label{sec:results}Results}

We conduct an in-depth analysis of contextual sets, focusing
on their structure and properties. Our aim is not to create
blueprints for experiments that prove their contextuality (because,
after so many experiments carried out so far, new ones would not
reveal anything new) or to use contextuality do disprove hidden
variable models (which are now barely regarded as relevant).
We also do not seek to derive new inequalities, as~existing tools
already allow us to demonstrate contextual sets efficiently.
Lastly, we do not aim to design BB84-like cryptographic contextual
protocols because they cannot provide a quantum advantage over the
noncontextual protocols; an~eavesdropper can easily ignore
conditions {\em(i)} and {\em(ii)} of Definition \ref{def:n-b}
and mimic quantum measurement~outcomes.

Instead, we are interested in the structure and properties of
contextual sets and their generation. As for the generation of
the sets we focus on several methods in dimensions up to eight
and on the methods whose complexity does not scale with dimension
to obtain sets in higher dimensions (here up to 32) by using sets
generated in lower dimensions by the previous methods. While
carrying out such a unification of methods, we also discuss the
realistic implementation of the sets themselves and their
applications, as well as the tools used to manipulate~them.
  
\subsection{\label{subsec:ks-non-ks}Kochen--Specker
  vs.~Non-Kochen--Specker MMPHs}

When considering implementations or applications of MMPHs,
we primarily focus on a set of quantum states represented by
vectors in $n$-dim Hilbert space, organized into $m$-tuples
(where $m\le n$) of mutually orthogonal vectors, with~at least
one hyperedge containing $m=n$ vertices. We impose the latter
requirement because we believe that it is plausible to assume
that the considered vertices (vectors, states, operators) fully
reside in the $n$-dim space within at least one of the
hyperedges, although~there are several exceptions to this
requirement in the literature, the~most notable one being the
2-dim pentagon in the 3-dim space
(\cite{pavicic-quantum-23}, Figure~5,~\cite{klyachko-08}). For~the
same reason, we will not try to equip $n$-dim MMPHs without
a coordinatization with a coordinatization from a higher-dimensional space, although~examples of such an embedding
do exist in the literature too
(\cite{pavicic-quantum-23}, Figure~4,~\cite{cabello-13a}).  

\begin{definition} An $n$-dim {\rm NBMMPH} with a
  coordinatization, in~which each hyperedge contains m = n
  vertices, is a Kochen--Specker {\rm (KS) MMPH}.
\label{def:KSdef}
\end{definition}

\begin{definition} An $n$-dim {\rm NBMMPH}  with a
  coordinatization, in~which at least one hyperedge contains
  $m<n$ vertices and at least one hyperedge contains
  m = n vertices, is a {\rm non-KS MMPH}.
\label{def:nonKSdef}
\end{definition}

Both KS and non-KS MMPHs are NBMMPHs and are therefore~contextual.

\subsection{\label{subsec:nonks}Generation of NBMMPHs}

To generate NBMMPHs, we make use of the following~methods.

\begin{itemize} 
\item{\bf M1}: Combines simple vector components to exhaust all
  possible collections of $n$ mutually orthogonal $n$-dim vectors.
  These vectors form {\em master} MMPHs, which consist of single
  or multiple MMPHs of varying sizes. Master MMPHs may or may not
  be NBMMPHs
~\cite{pm-entropy18,pavicic-pra-22}.
\item{\bf M2}: Automated dropping of hyperedges contained in
  masters found by {\bf M1, M6--8} or~by some other method in
  the literature; they serve to generate {\em classes} of
  smaller MMPHs massively
~\cite{pm-entropy18,pavicic-pra-22}.
\item{\bf M3}: Automated dropping of vertices contained in single
  hyperedges (multiplicity $m=1$) of either NBMMPHs or BMMPHs and
  the possible subsequent stripping of their hyperedges
~\cite{pavicic-quantum-23}. The~obtained smaller MMPHs are often
  NBMMPH, although~never KS.  
\item{\bf M4}: Automated random addition of hyperedges to MMPHs
  to obtain larger ones, which then generates smaller KS MMPHs
  through the random removal of hyperedges again.
\item{\bf M5}: Deleting vertices in either an NBMMPH
  or a BMMPH until a non-KS NBMMPH is reached, if~any.
\item{\bf M6}: {Downward generation of NBMMPHs from fortuitously
  or intuitively found connections of KS MMPHs with polytopes,
  Pauli operators, or~space \mbox{symmetries~\cite{pavicic-pra-17,pwma-19,pavicic-quantum-23}.}}
\item{\bf M7}: Generation of KS MMPHs in higher dimensions from
  the ones in smaller dimensions~\cite{zimba-penrose,cabell-garc96,cabell-est-05,matsuno-07}.
\item{\bf M8}: Generation of KS MMPHs in higher dimensions by
  dimensional upscaling, whose complexity does not scale with
  the dimension~\cite{pw-23a}.
\end{itemize}

We combine all of these methods to obtain an arbitrary number of
NBMMPHs in an arbitrary dimension and to generate KS and non-KS
MMPHs in dimensions up to 32 below.

To familiarize ourselves with non-KS MMPHs, in~Figure~\ref{fig:non-KS-fig}, we visualize the procedures of
obtaining them from two types of KS MMPHs and (noncontextual)
BMMPHs by means of our~methods. 

\begin{figure}[H]
    \includegraphics[width=\linewidth]{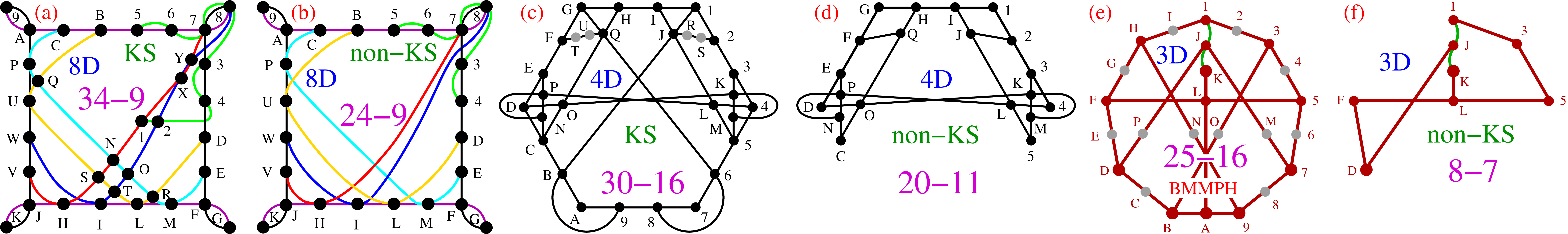}
\caption{Obtaining 
 non-KS MMPHs via
    three different methods (colored lines represent hyperedges):
    (\textbf{a}) 8-dim KS MMPH obtained via  {\bf M1} or
    {\bf M6}; (\textbf{b}) 8-dim
    non-KS MMPH $\overline{\rm subhypergraph}$\break
    of (a) obtained by {\bf M5} (deleting vertices inside
    the rectangle loop);
    (\textbf{c}) 4-dim KS MMPH obtained via  {\bf M1} or {\bf M6};
    (\textbf{d}) 4-dim non-KS MMPH $\overline{\rm subhypergraph}$
    obtained from (c) by {\bf M2} and {\bf M3} (successive
    deletions of $m=1$ (grey) vertices and hyperedges so that, at
    each step, the resulting MMPH is a non-KS MMPH---until
    we reach the smallest critical non-KS without $m=1$ vertices);
    (\textbf{e}) 3-dim (noncontextual) BMMPH subhypergraph obtained
    by {\bf M1} and {\bf M2};
    (\textbf{f}) 3-dim non-KS MMPH $\overline{\rm subhypergraph}$
    obtained from (e) by {\bf M2} and {\bf M3} (successive
    deletions of vertices and hyperedges so that, at each
    step, the resulting MMPH is a non-KS MMPH---until we
    reach a critical non-KS without $m=1$ vertices);
    strings and coordinatizations of
    (\textbf{a},\textbf{c},\textbf{e}) are given in
    the Appendix 
    \ref{app:0} since they were not given elsewhere; strings and
    coordinatizations of (\textbf{b},\textbf{d},\textbf{f}) can
    be derived from those of (\textbf{a},\textbf{c},\textbf{e}).}
\label{fig:non-KS-fig}
\end{figure}

The  8-dim KS MMPH shown in Figure~\ref{fig:non-KS-fig}(a) is 
obtained via {\bf M6} from the 120--2025 master but can also
be obtained from the 3280--1361376 master directly generated
from the $\{0,\pm 1\}$ vector components, i.e.,~via {\bf M1}
\cite{pm-entropy18}. The~ 4-dim KS MMPH, in 
Figure~\ref{fig:non-KS-fig}(c), is obtained via {\bf M6} from
the 60--105 Pauli operator master~\cite{pavicic-pra-17}, but~can also be obtained from the 156--249 master directly
generated from the $\{0,\pm 1,\pm i\}$ vector components,
i.e., via {\bf M1} \cite{pm-entropy18}. The~3-dim
noncontextual BMMPH, in Figure~\ref{fig:non-KS-fig}e, is obtained
via {\bf M6} from the Peres master, but~can also be obtained
from the 81--52 master directly generated
from the $\{0,\pm 1,\sqrt{2},\pm 3\}$ vector components,
i.e., via {\bf M1} \cite{pavicic-pra-22}. 

One can easily verify by hand that the MMPHs shown in 
Figure~\ref{fig:non-KS-fig}(b,d,f) violate the conditions
{\em (i)} and {\em (ii)} in Definition~\ref{def:n-b},
i.e.,~that they are NBMMPHs and that they, by~having
coordinatizations, are, by~Definition~\ref{def:nonKSdef},
non-KS~MMPHs. 

\subsection{\label{subsec:ineq}NBMMPHs
  vs.~Operators and States---The Inequalities}

In the literature, contextual sets, mostly KS ones, have often
been formulated by means of operators (mostly projectors) and
states, especially together with a proposed~implementation.

Recently, it was shown that ``a proof of the Kochen–Specker
theorem can always be converted to a state-independent [operator]
noncontextuality inequality'' \cite{yu-tong-14,yu-tong-15}. 
Proofs of the KS theorem are KS sets, and the result actually
means that, for any KS set with a coordinatization, i.e.,~for any
KS MMPH with a coordinatization, we can find quantum operators
whose particular expressions would give the same result regardless of the states that they are applied to. More precisely,
one can form expectation values that correspond
to the measurement values of these expressions and
differ from the expectation values of the assumed classical
counterparts. The~inequality between these two expectation
values, quantum and classical, is called a noncontextuality
inequality.

Obviously, this~result does not apply to KS sets without a
coordinatization, and the question arises of whether we can
form similar discriminators, i.e.,~some inequalities for
hypergraphs without reference to coordinatization, even
when the hypergraphs do possess it. To~arrive at an answer
to this question, we have to introduce a few definitions.
But before doing so, we give an example of a
state-independent operator setup of a KS set
and its noncontextuality inequality.

In ref.~(\cite{cabello-08}, Equation~(2)), the following 4-dim
operators are defined
\begin{eqnarray}
  A_{ij}=2|v_{ij}\rangle\langle v_{ij}|-I
\label{eq:a-cab}
\end{eqnarray}
through the vector coordinatization of the 4-dim KS 18-9 MMPH
shown in (\cite{cabello-08}, Figure~1), e.g.,
\begin{eqnarray}
  |v_{12}\rangle=(1,0,0,0); \ \
|v_{16}\rangle=\frac{1}{\sqrt{2}}(0,0,1,-1); \ \
\ \ \dots\ \
|v_{58}\rangle=\frac{1}{\sqrt{2}}(1,0,-1,0); \ \ \dots 
\label{eq:v-cab}
\end{eqnarray}
The reader can find all vectors in (\cite{cabello-08}, Figure~1).
We do not show all vectors here because we only need one
of them for our Equation~(\ref{eq:cab-non-igen}) below.
In ref.~(\cite{cabello-08}, Equation~(2)), it is then claimed that
\begin{eqnarray}
  A_{ij}A_{ik}A_{il}A_{im}=-I,
\label{eq:cab-s-i}
\end{eqnarray}
where $i=1,\dots,9$ ranges over 9 hyperedges of the 18-9 KS MMPH,
and $j,k,l,m$ ranges over the vertices within each of the hyperedges,
which yields (full $\Omega$ is given in (\cite{cabello-08}, Equation~(1)))
\begin{eqnarray}
  -\langle A_{12} A_{16}A_{17} A_{18}\rangle-\dots-\langle A_{29}
  A_{39} A_{59} A_{69}\rangle = \langle\Omega\rangle,
\label{eq:a-cab2}
\end{eqnarray}
would take us to a ``state-independent''  noncontextuality
inequality. Specifically, it is argued that a quantum
interpretation of $A_{ij}$ that relies on ``measurements of
subsets of compatible observables on different subensembles
prepared in the same state'' yields $\langle\Omega\rangle=9$,
while a  classical interpretation of $A_{ij}$ (which assigns
$-$1 or 1 to them) yields a $\langle\Omega\rangle\le 7$
inequality. The~problem with this argumentation is that each $A_{ij}$ itself is not ``state-independent'', e.g.,
\begin{eqnarray}
  A_{12}|v_{58}\rangle= A_{12}\frac{1}{\sqrt{2}}(1,0,-1,0)=
  -\frac{1}{\sqrt{2}}(0,0,1,0),
\label{eq:cab-non-igen}
\end{eqnarray}
i.e., $|v_{58}\rangle$ is not an eigenvector of $A_{12}$.
This is important because we must have correspondence
between the quantum $A_{ij}$ and classical $A_{ij}$---they both
have to be measurable and the measurement outcome should be
$\pm 1$ for both. This is the essence of the KS~theorem.

A similar problem exists with most of the proposed
state-independent operator assessments of contextual sets. They
all involve products of operators that assume the passing of the
states through several devices and therefore prevent them from
being measured at each of them separately~\cite{yu-tong-14,yu-oh-12,zu-wang-duan-12,qu-24,pavicic-quantum-23},
while they are expected to yield a noncontextuality inequality.
However, contextual sets actually do not require operator
representation, as~there is always an NBMMPH formulation
for any of the sets and this formulation only requires measurements
of bare states for any implementation and an algorithmic check
of the data for contextuality~verification.

Another approach to contextuality is provided by the so-called
true-implies-false (TIF) and true-implies-true (TIT) sets
~\cite{cabello-svozil-18,svozil-21}, some of whose diagrams are
presented in Figure~\ref{fig:tif}. For~instance, in~TIF diagrams,
if we assign value 1 (TRUE) to {\tt A} and proceed to assign
values so as to satisfy conditions {\em(i)} and {\em(ii)} from
Definition~\ref{def:n-b} along all lines simultaneously, we
find that {\tt B} must be assigned value 0 (FALSE). 

One can verify that all TIF and TIT diagrams
from~\cite{cabello-svozil-18,svozil-21} are non-KS~NBMMPHs.

\begin{figure}[H]
    \includegraphics[width=\textwidth]{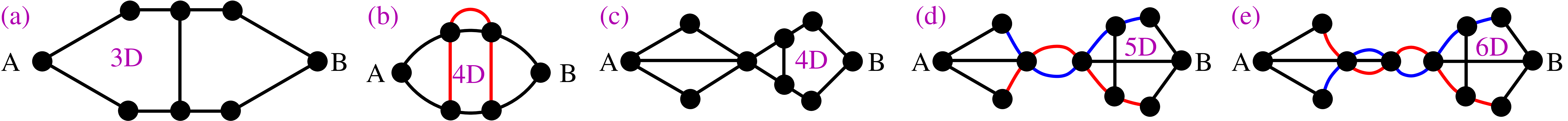}
  \caption{TIF diagrams according to
    figures~from~\cite{cabello-svozil-18} (except (\textbf{b}))
    (colored lines represent hyperedges):
    (\textbf{a}) (\cite{cabello-svozil-18}, Figure~1a),
    cf.~Figure~\ref{fig:bug};
    (\textbf{b}) 4-dim NBMMPH;
    (\textbf{c}) (\cite{cabello-svozil-18}, Figure~4a);
    (\textbf{d}) (\cite{cabello-svozil-18}, Figure~5a);
    (\textbf{e}) (\cite{cabello-svozil-18}, Figure~7a);
    (\textbf{a},\textbf{c},\textbf{d},\textbf{e}) are non-KS MMPHs;
    (\textbf{b}) would be a KS MMPH
    if it had a coordinatization, but~it does not, so it
    is simply a NBMMPH.}
\label{fig:tif}
\end{figure}

\medskip 
To obtain a better insight into the properties of MMPHs,
let us proceed.

\begin{definition} The
{\bf MMPH Quantum Hypergraph Index} $HI_q$ is the sum of the weighted
probabilities of all vertices of an $n$-dim $k$-$l$ {\rm MMPH}
measured repeatedly in all hyperedges that they belong to,  whenever
their multiplicity is $>1$.
\label{def:qh-i}
\end{definition}

\begin{lemma}{\bf Vertex-Hyperedge Lemma.}
  For any $n$-dim $k$-$l$ {\rm MMPH} in which each hyperedge
  contains $n$ vertices, the following holds:
\begin{eqnarray}
HI_q=\sum^k_{i=1}\frac{m(i)}{n}=l.
\label{eq:theorem}
\end{eqnarray}
 In general, for~any $n$-dim $k$-$l$ {\rm MMPH} with
 $\kappa(j)$ considered vertices in the $j$-th hyperedge,
 $j=1,\dots,l$,  the~following holds:
\begin{eqnarray}
  HI_q=\sum^l_{j=1}\sum^{\kappa(j)}_{\lambda=1}p(j,\lambda)=l,
\label{eq:theorem-g}
\end{eqnarray}
where $\kappa(j)$ is the number of vertices in a hyperedge $j$
and $p(j,\lambda)=\frac{1}{\kappa(j)}$ is the probability that
a state of a system corresponding to one of the vertices would
be detected when the hyperedge $j$ is being measured.
\label{th:theorem}
\end{lemma}

\begin{proof} Equation~(\ref{eq:theorem}) is equivalent to a
  generalized Handshake Lemma for Hypergraphs 
  (\cite{melnikov-98}, Exercise 11.1.3.a). The~proof is given
  in~\cite{pavicic-quantum-23}.

  To prove Equation~(\ref{eq:theorem-g}), we simply note that
  $\sum^{\kappa(j)}_{\lambda=1}p(j,\lambda)=1$ for any $j$.
\end{proof}

\begin{definition} \label{def:iin} The {\bf v-inequality.}
  An {\rm MMPH} vertex inequality or simply {v-inequality}
  is defined as\vspace{-3pt}
\begin{eqnarray}
HI_{cm}\le HI_{cM}\le HI^m_{cM}<HI_q=l.
\label{eq:i-ineq}
\end{eqnarray}
\end{definition}

\begin{lemma}
  All $n$-dim {\rm NBMMPH}s satisfy the v-inequality.
\label{lemma:v}
\end{lemma}
\begin{proof}
  In an NBMMPH, the maximal number of hyperedges
  that contain `1' must be smaller than the total number of
  hyperedges $l$ by definition.
\end{proof}

\begin{lemma}\label{lemma:ein} The {\bf e-inequality.}
  $l_{cM}$ ($l_{cm}$) satisfies the following hyper{\bf e}(d)ge
  inequality or simply {\bf e$_{Max}$-inequality}
({\bf e$_{min}$-inequality}):
\begin{eqnarray}
  l_{cM}<l \qquad  (l_{cm}<l).
\label{eq:e-ineq}
\end{eqnarray}
\end{lemma}

\begin{proof}
They are noncontextuality inequalities simply because
$l_{cm}=l_{cM}=l$ for all binary MMPHs by definition.
\end{proof}

Apparently, $l_{cm}$ is the ``rank of contextuality'' \cite{horod-22},
introduced as a quantifier of contextuality for hypergraphs.
Both e$_{Max}$- and e$_{min}$-inequalities can
simply be called~e-inequalities.

\begin{lemma}
  All $n$-dim non-binary {\rm MMPH}s satisfy the e-inequalities.
  \label{lemma:e}
\end{lemma}

\begin{proof}
  For KS MMPHs, it follows directly from the KS theorem,
  since both the maximal and minimal number of
  hyperedges that contain 1 must be smaller than the total
  number of hyperedges $l$. For~non-KS NBMMPHs, it follows
  from Definition~\ref{def:n-b} and its conditions {\em(i)}
  and {\em(ii)} in the same way. 
    \end{proof}

\begin{definition}{\bf Fractional independence (packing) number
    $\alpha^*({\cal H})$ (LP)} of an {\rm MMPH}\break
  ${\cal H}(k$-$l)$ is the optimum value of the following linear 
  programming problem $LP=LP({\cal H})$

  \smallskip
  \qquad (LP) Maximize $\sum_{v\in V}x(v)$

\medskip
\qquad\phantom{(LP)\ }subject to $\sum_{v\in e}\le 1$,
$\forall e\in E$

\medskip
\qquad\phantom{(LP)subject to\ $\>$}$x(v)\in[0,1]$,
$\forall v\in V$
\label{def:alpha-star-LP}
\end{definition}

\begin{definition}{\bf Fractional independence number}
    $\alpha^*({\cal H})$ is defined as (Wolfram
~\cite{wolfram-alpha-24}):
\begin{eqnarray}\label{eq:alpha-star-Wolfram}
  \alpha^*({\cal H})
    =\max_{\sum_{v\in e}\le 1, \forall e\in E}\sum_{v\in V}x(v);
\quad x(v)\in [0,1].
  \label{def:alpha-star-Wolfram}
\end{eqnarray}
  \end{definition}

If the v- and e-inequalities are satisfied, an~MMPH would be
contextual. If~not, it would not. Thus, the~v- and e-inequalities
are noncontextuality inequalities.

On~the other hand, there is
also the following inequality (the so-called
$\alpha$-inequality):
\begin{eqnarray}\label{eq:alpha-cabelllo}
\alpha({\cal H})\le\alpha^*({\cal H}),   
\end{eqnarray}
where $\alpha({\cal H})$ is defined by Definition~\ref{def:alpha}
and $\alpha^*({\cal H})$ by Definitions~\ref{def:alpha-star-LP} and
\ref{def:alpha-star-Wolfram}, which
is claimed to be a noncontextuality inequality too
(\cite{cabello-severini-winter-14}, Results 1 and 2), although it
should hold for any (hyper)graph (\cite{gro-lovasz-schr-81}, p.~192),
either non-binary or binary, i.e.,~for both NBMMPHs and
BMMPHs. 

Equation~(\ref{eq:alpha-cabelllo}) certainly holds for general
(hyper)graphs~\cite{gro-lovasz-schr-81} (called the GLS
inequality) where $x(v)$ is a free variable, i.e.,~where
$x(v)=w_v p(v)$, where $p(v)$ is the probability of detecting
vertex $v$ and $w_v$ is the weight of this probability
~\cite{cabello-severini-winter-14}. In~the latter reference,
it is stated that finding $\alpha^*$ is NP-hard. This is
correct for the general GLS, but~is it so for quantum
measurements of MMPHs? 

\begin{wrapfigure}{r}{0.25\textwidth}
\vspace{-13pt}
  \begin{center}
    \includegraphics[width=0.2\textwidth]{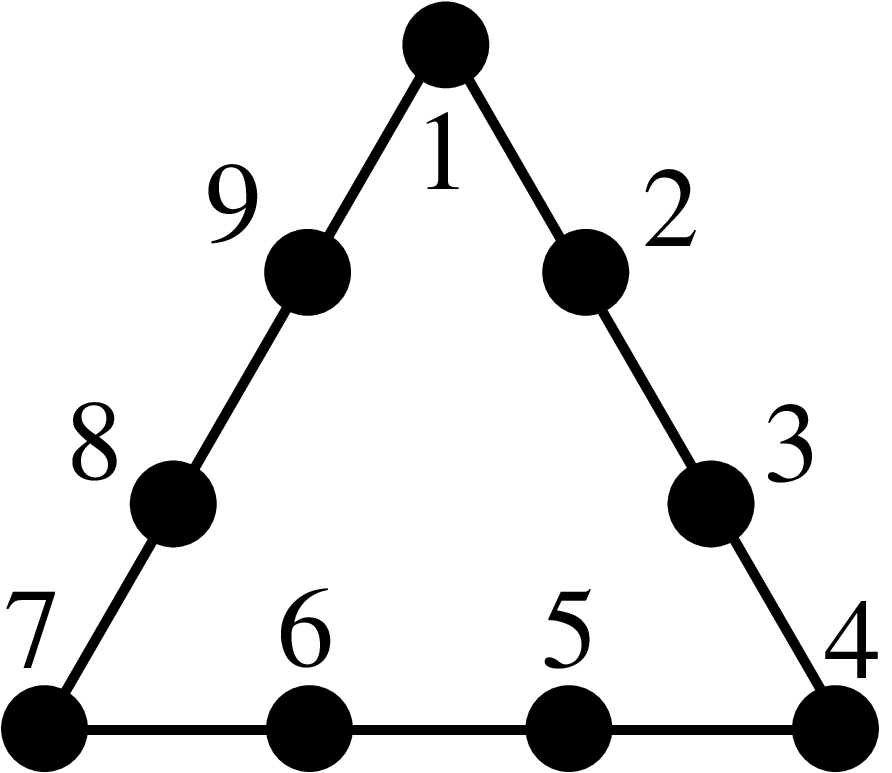}
  \end{center}
\vspace{-15pt}
  \caption{BMMPH 9-3.}
\label{fig:alpha}
\end{wrapfigure}

Let us examine MMPH 9-3 given in Figure~\ref{fig:alpha}.
For a free $x$, we have

{LP[\{$-$1,$-$1,$-$1,$-$1,$-$1,$-$1,$-$1,$-$1,$-$1\},\{\{1,1,1,1,0,0,0,0,0\},\linebreak \{0,0,0,1,1,1,1,0,0\},\{1,0,0,0,0,0,1,1,1\}\},\{\{1,$-$1\},\{1,$-$1\},\{1,$-$1\}\}].}

(LP is carried in {\em Mathematica} (LinearProgramming[]).
Vectors/vertices are given as follows: $\{1,1,1,1,0,0,0,0,0\}$
stands for the first hyperedge, {\tt 1234}, $\{0,0,0,1,1,1,1,0,0\}$
for the second one, {\tt 4567}, and~$\{1,0,0,0,0,0,1,1,1\}$ for the
third one, {\tt 7891}).

Out: =\{0,1,0,0,1,0,0,1,0\}, i.e, $\alpha^*=3$.

Since $\alpha=3$, inequality (\ref{eq:alpha-cabelllo}) is~satisfied.

However, for~$x=p=\frac{1}{4}$, i.e., for quantum measurements,
we~obtain

{LP[\{$-$1,$-$1,$-$1,$-$1,$-$1,$-$1,$-$1,$-$1,$-$1\},\{\{1,1,1,1,0,0,0,0,0\},\{0,0,0,1,1,1,1,0,0\},\{1,0,0,0,0,0,1,1,1\}\},\linebreak\{\{1,$-$1\},\{1,$-$1\},\{1,$-$1\}\}, \{\{$\frac{1}{4}$,1\},\{$\frac{1}{4}$,1\},\{$\frac{1}{4}$,1\},\{$\frac{1}{4}$,1\},\{$\frac{1}{4}$,1\},\{$\frac{1}{4}$,1\},\{$\frac{1}{4}$,1\},\{$\frac{1}{4}$,1\},\{$\frac{1}{4}$,1\}\!\}\!].}

Out: =\{$\frac{1}{4}$,$\frac{1}{4}$,$\frac{1}{4}$,$\frac{1}{4}$,$\frac{1}{4}$,$\frac{1}{4}$,$\frac{1}{4}$,$\frac{1}{4}$,$\frac{1}{4}$\},

\noindent i.e., $\alpha^*=\frac{9}{4}=2.25$, which violates
inequality (\ref{eq:alpha-cabelllo}).

See also the violations presented in the figure in
Section~\ref{subsubsec:3D}.

In fact, such a regular/random distribution of, say, $n$-dim
spin-$\frac{n-1}{2}$ MMPH systems exiting through the ports
(vertices) of each of their gates/hyperedges (e.g., in a
Stern--Gerlach device) is a well-known distribution of the spin
projections of states of such systems. For~systems composed of
of several subsystems (e.g., several qubits, qutrits, and~so on),
a description is somewhat more complicated, but,~in general,
we can say that the systems satisfy the following~postulate. 

\begin{resultttt}\label{postulate} Quantum 
  systems generated in an unknown (unprepared) pure state in an
  apparatus (e.g., a~generalized Stern--Gerlach one), when exiting
  from it through one of the out-ports (channels, vertices) of
  their gates/hyperedges, have an equal probability of being detected
  on their exit (\cite{feynmanIII}, Section~5-1).
\end{resultttt}

Hence, Equation~(\ref{eq:alpha-cabelllo}) does not generally apply to
quantum measurements, and the calculation of $\alpha^*$ is not of NP
but of linear complexity. In fact, Equation~(\ref{eq:alpha-cabelllo})
fails for many arbitrary quantum measurements, as shown~below.

Measurements of a $k$-$l$ set are carried out on 
hyperedges---hyperedge by hyperedge---and each hyperedge yields
a single detection (click) corresponding to one of $n$ vertices
(vectors, states) contained in the hyperedges with a probability
of $\frac{1}{n}$. This means that, for MMPHs whose hyperedges
all contain $n$ vertices, one builds the following statistics.

\begin{definition}\label{def:raw}
  {\bf Raw data statistics} for {\rm MMPH}s whose all hyperedges
  contain $n$ vertices (often adopted in the literature, e.g.,
  {\rm(\cite{d-ambrosio-cabello-13}}, Eq.~(2),
  {\rm\cite{yu-oh-12}}, lines under Eq.~(2)), etc.{\rm)},
  consist of assigning $\frac{1}{n}$ probability to each of
  the $k$ vertices contained in the hypergraph (see
  Definition~\ref{def:n-b}), independently of whether the vertices
  appear in just one hyperedge or in two or more of them. 
\end{definition}

Such statistics do not appear to be satisfactory, though. Consider
again 9-3 in Figure~\ref{fig:alpha}. The~raw data statistics give
us the sum of probabilities $\frac{9}{4}$. However,~the vertices
{\tt 1}, {\tt 4}, and~{\tt 7} share two hyperedges and have
multiplicity $m=2$. Thus, we actually have the sum of such
calibrated probabilities being equal to
$\frac{6}{4}+\frac{2\times 3}{4}=3$, i.e.,~we have a calibrated
$\alpha^*=3$, which satisfies the inequality
(\ref{eq:alpha-cabelllo}). We denote such a calibrated
$\alpha^*$ as $\alpha^*_p$ and introduce the following more
appropriate statistics:

\begin{definition}\label{def:post}
{\bf Postprocessed MMPH data statistics} are statistics
  for~which 
\begin{enumerate}
\item vertex \textquoteleft$v$\textquoteright\ might share $m(v)$
  hyperedges;
\item measurements are performed on $n$ vertices
  $v(j)$ contained in hyperedges \textquoteleft$j$\textquoteright,
  $j=1,\dots,l$;
\item the outcomes of measurements carried out on particular
  vertices $v(j)$ in a particular hyperedge $j$ might be dropped
  out of consideration, leaving us with $\kappa(j)$ vertices in
  the hyperedge $j$;
\item the probability of obtaining measurement data for each vertex
  within a hyperedge, after~discarding the data for $n-\kappa(j)$
  dropped vertices, is  $\frac{1}{\kappa(j)}$;
\item the sum of all probabilities is, according to
  Equation~{\rm(\ref{eq:theorem-g})}, equal to the size of the
  hypergraph, i.e.,~to the number of its hyperedges $l$.    
\end{enumerate}
\end{definition}

Such statistics yield the following~theorem. 

\begin{Xeorem}Let variables $x(v)$ from
  Definition~{\rm\ref{def:alpha-star-LP}}
  be the probabilities $p(v)$, $v\in V$ of detecting an event by
  {\rm YES-NO} measurements at one of the out-ports (vertices)
  contained within a hyperedge of an $n$-dim {\rm MMPH}
  ${\cal H}(V,E)={\cal H}(k$-$l)$. Each of $E_j\in E$,
  $j=1,\dots,l$ hyperedges (gates) contains $n$ vertices. The~  following holds:
\begin{eqnarray} \sum_{v\in {E_j}}p(v)\le 1, j=1,\dots,l.
    \label{eq:sum-edge}
  \end{eqnarray}
  They also satisfy the following:

  (a) Under the {raw data statistics} (Definition~{\rm\ref{def:raw}})
  assumption
  {\rm~\cite{cabello-severini-winter-14,magic-14,amaral-cunha-18}},
  the sum of all prob-\break abilities is
\begin{eqnarray} \sum_{v=1}^kp(v)=\frac{k}{n}
   =\alpha^*(k{\text{-}}l)=\alpha^*({\cal H}).
  \label{eq:alpha-star}
\end{eqnarray}

(b) Under the {postprocessed data statistics}
  (Definition~{\rm\ref{def:post}}) assumption, i.e.,~under the
  assumption that every vertex $v$
  within an {\rm MMPH} has $\frac{m(v)}{n}$ probability of being
  detected, the~sum of all probabilities, according to the
  vertex-hyperedge Lemma \ref{th:theorem},  is
\begin{eqnarray} \sum_{v=1}^kp(v)
    =\sum_{v=1}^k\frac{m(v)}{n}
    =l=\alpha_p^*(k{\text{-}}l)=\alpha_p^*({\cal H})
  \label{eq:alpha-star-b}
\end{eqnarray}
where {\boldmath{$\alpha^*_p$}} is called the {\em postprocessed
  quantum fractional independence number}.
\end{Xeorem}

\begin{wrapfigure}{r}{0.45\textwidth}
  \begin{center}
\vspace{-18pt}
    \includegraphics[width=0.25\textwidth]{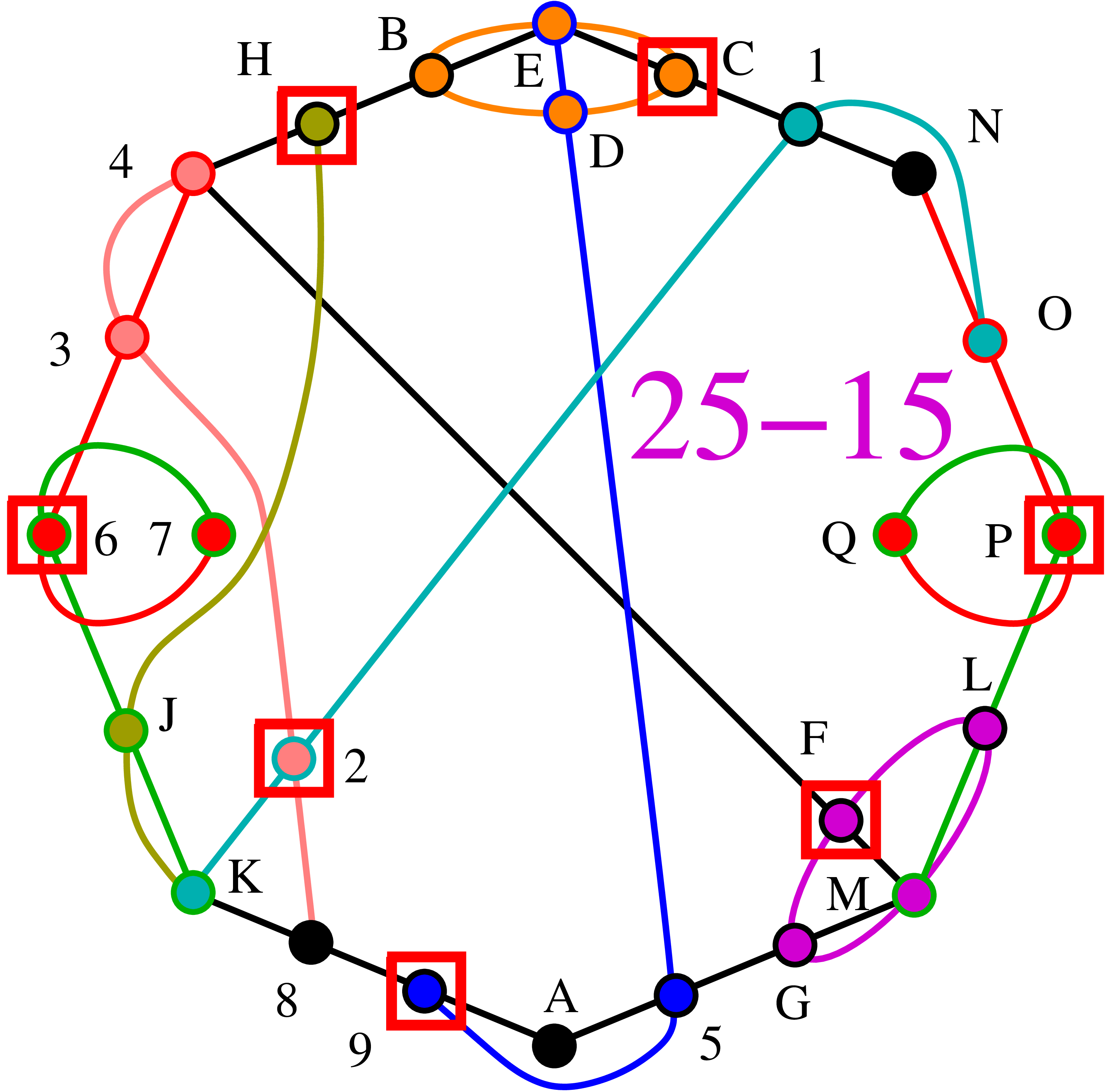}  
  \end{center}
\vspace{-15pt}
\caption{4-dim non-KS NBMMPH 25-15, a
  $\overline{\rm subhypergraph}$ of 26-15
  (\cite{pavicic-quantum-23}, Figure~6b);
  $\alpha=7>\alpha^*=\frac{25}{4}=6.25$; vertices
  contributing to $\alpha$ are red-squared; string
  and coordinatization are given in Appendix \ref{app:2}.}
\label{fig:alpha-star}
\vspace{-15pt}
\end{wrapfigure}

Statement (a) implies that the
{\boldmath{$\alpha$}}{\bf-inequality}
\begin{eqnarray} 
 HI_{cM}=\alpha({\cal H})\le \alpha^*({\cal H})=\frac{k}{n}
    \label{eq:alpha-alpha}
\end{eqnarray}
does not always hold for quantum mechanical measurements whose
probabilities of detection 
  within each hyperedge satisfy the
  condition given by Equation~(\ref{eq:sum-edge}), i.e.,~under
  the Quantum Indeterminacy Postulate \ref{postulate}, as~shown
  above in the LP analysis of {\rm MMPH} 9-3
  in~Figure~{\rm\ref{fig:alpha-star}} and in arbitrarily many
  other NBMMPHs in any dimension (see
  (\cite{pavicic-quantum-23}, Figure~6, Tables~2, 6, and 7).
  Note that $\alpha^*_r$ there (\cite{pavicic-quantum-23}) is
  equal to $\alpha^*$ here, according to the Quantum Indeterminacy
  Postulate \ref{postulate}. It is, therefore, not a reliable
  discriminator of contextual~sets.

Statement (b) implies that the
  {\boldmath{$\alpha^*_p$}}{\bf-inequality}
\begin{eqnarray} 
  HI_{cM}=\alpha({\cal H})< \alpha_p^*({\cal H})=l=HI_q,
    \label{eq:alpha-alpha-b}
\end{eqnarray}
which follows from the vertex-hyperedge lemma
(Equation~{\rm(\ref{eq:theorem-g})}), is another form of the
v-inequality (Equation~{\rm(\ref{eq:i-ineq})})
and that it is therefore a noncontextuality inequality and a
reliable discriminator of contextual sets. 
\label{th:alpha-star}

Taken together, the~MMPH e-inequality (Equation~(\ref{eq:e-ineq})),
v-inequality (Equation~(\ref{eq:i-ineq})),
and~$\alpha^*_p$-inequality (Equation~(\ref{eq:alpha-alpha-b})) are
genuine noncontextual inequalities and reliable discriminators
of contextual sets, while the $\alpha$-inequality
(Equations~(\ref{eq:alpha-cabelllo},\ref{eq:alpha-alpha-b})) is not,
even though it is frequently presented as if it were,
in~the~literature.

\subsection{\label{subsec:gen}  Generations of KS
  and Non-KS MMPHs in Dimensions 3 to 32}

 We provided a fairly exhaustive presentation of the generation of
 both KS and non-KS MMPHs in dimensions 3 to 8 in~\cite{pavicic-pra-17,pavicic-entropy-19,pavicic-entropy-23,pavicic-quantum-23,pm-entropy18,pavicic-pra-22,pwma-19,pw-23a}.
 Here, we shall only present some new results and refer
to particular previous findings in tables for the sake of
completeness. 

\subsubsection{\label{subsubsec:3D}  3-dim {\rm MMPH}s}

The smallest 3-dim non-KS MMPH (pentagon) \cite{klyachko-08}
that we analyzed in~\cite{pavicic-quantum-23} does not satisfy
the basic requirement that all $n$-dim MMPHs in this paper
should satisfy, i.e., that at least one of the hyperedges
should have $n$ vertices, meaning that at least one triple of
orthogonal vectors should live in a 3-dim space. Thus,
for~example, the~smallest non-KS NBMMPH 7-7, obtained by
{\bf M4-5}, smaller than the ``bug'' and shown in
Figure~\ref{fig:3D}(a), does contain one hyperedge with
three~vertices. 
\begin{figure}[H]
  \includegraphics[width=\linewidth]{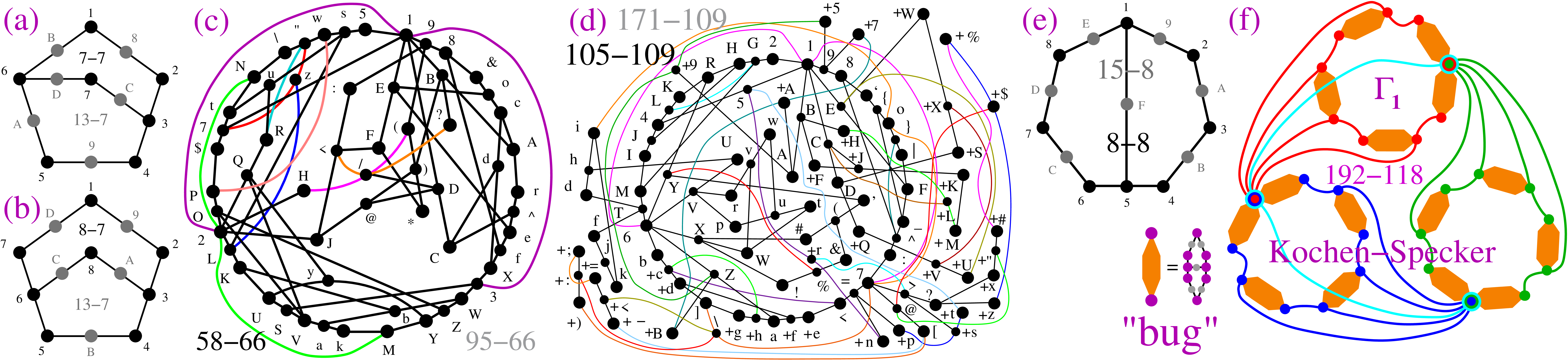}
  \caption{3-dim MMPHs. (\textbf{a},\textbf{b},\textbf{e})
    Three smallest non-KS
  NBMMPHs; none can be substituted for the ``bug'' in $\Gamma_1$ in
  (\textbf{f})---see text; (\textbf{c}) the smallest critical KS
  MMPH 95-66 obtained from golden ratio components (from 597-378
  master) presented in the text (shown without $m=1$ vertices);
  (\textbf{d}) a much larger critical KS MMPH from the same master
  (shown without $m=1$ vertices); (\textbf{e}) 8-8 non-KS
  NBMMPH---see text; (\textbf{f}) a proper hypergraph encoding of
  the original Kochen and Specker design---see
  (\cite{pavicic-quantum-23}, Figure~11); cf.~improper encoding
  in (\cite{budroni-cabello-rmp-22}, Figure~1)---see text;
  $m=1$ vertices in (a,b,e) are depicted as grey dots.}
\label{fig:3D}
\end{figure}
The coordination of the filled 7-7 MMPH
({\tt 182,234,495,5A6,6B1,3C7,6D7.}) obtained by our programs is
1 = (1,$-$1,1), 2 = (1,0,$-$1), 3 = (1,0,1), 4 = (0,1,0),
5 = (1,0,0), 6 = (0,1,1),\linebreak 7 = (1,1,$-$1), 8 = (1,2,1),
9 = (0,0,1), A = (0,1,$-$1), B = (2,1,$-$1), C = ($-$1,2,1),
D = (2,$-$1,1).\linebreak

The next non-KS NBMMPHs obtained by the
same methods are the 8-7 shown in Figure~\ref{fig:3D}d and the
``bug'' shown in Figure~\ref{fig:bug}.

These three MMPHs are interesting because, together with
pentagons, hexagons, and~heptagons, they prove that the ``bug''
plays a specific role in the structure of the original
Kochen--Specker proposal, whose proper 192-118 MMP rendering is
shown in Figure~\ref{fig:3D}(f) (cf.~(\cite{budroni-cabello-rmp-22},
Non-Figure 1)---where words are offered as unusual substitutes for
hyperedges). If~we substituted the aforementioned MMPHs for the
bugs in the 192-118, they would transform it into a~BMMPH. 

The 8-8 shown in Figure~\ref{fig:3D}(e) is obtained by {\bf M2-3}  
from the 33--50 non-KS, which is a $\overline{\rm subhypergraph}$
(black dots) of the  69-50 shown in
\cite[Figure~10e]{pavicic-quantum-23}.

Thousands of other KS MMPHs in the real and complex 3-dim Hilbert
space are presented in~\cite{pavicic-pra-22}, while the non-KS ones
derived from them are presented in~\cite{pavicic-entropy-23}.

Coordinatizations that we did not explore previously are those
derived from variants of the {\em golden ratio}
$\phi=(1+\sqrt{5})/2$. For~instance, the~following vector 
components:
$\{0, \pm1, \pm2^{\frac{1}{2}}\phi^{\frac{3}{2}},
\pm2^{\frac{1}{2}}\phi^{-\frac{1}{2}},
5^{\frac{1}{4}}\phi^{\frac{3}{2}},\linebreak
5^{\frac{1}{4}}\phi^{-\frac{3}{2}}, \pm2^{-1}\phi^{-1},
2^{-1}5^{\frac{1}{4}}\phi^{\frac{1}{2}},
\pm2^{\frac{1}{2}}\phi^{-\frac{5}{2}},
2^{-1}5^{\frac{1}{4}}\phi^{-\frac{5}{2}}\}$.
They generate the 597-358 master, which in turn generates over
250 criticals, the~smallest of which is shown in
Figure~\ref{fig:3D}(c) and one of the largest in
Figure~\ref{fig:3D}(d). Note that, in
\cite[Supp.~Material, p.~3]{pavicic-pra-22},
we generated vectors, also partly based on the golden ratio,
for~the original Kochen and Specker's design of their 192-118,
since the equation they gave in \cite{koch-speck} does
not provide us with any definitive~coordinatization (it
contains two arbitrary parameters).

Table \ref{T:3D} gives an overview of the 3-dim MMPHs discussed above.

\begin{table}[H]
\label{T:3D}
\setlength{\tabcolsep}{3.8pt}{\begin{tabular}{|c|cccccc|}
    \hline
  \multirow{3}{*}{\textbf{dim}} & &
  \textbf{No.~of}  & \multirow{3}{*}{\textbf{Methods}} 
  & \multirow{2}*{\textbf{Smallest KS and}}
  & \multirow{2}*{\textbf{Vector}} & \multirow{3}{*}{\textbf{References}} \\
&  \textbf{Master}   &  \textbf{Non-Isom} & & \multirow{2}*{\textbf{Non-KS Criticals}} &
      \multirow{2}*{\textbf{Components}} & \\
&        & \textbf{Criticals}    &  &     & & \\
  \hline
  \multirow{6}{*}{3-dim}& 13-7 $^\dag$ & 1 &
  {\bf M4-5}& 7-7 $\ddag$ & $\{0,\pm 1,2\}$ & see text\\
  & 13-7 $^\dag$ & 1 &
 {\bf M4-5}& 8-7 $\ddag$ & $\{0,\pm 1,2\}$ & see text\\
 & 97-64  & 1 & {\bf M1-2} &
 49-36 * (33-36 $\ddag$) & $\{0,\pm 1,\pm 2,5\}$ &
  \cite{pavicic-quantum-23} \\
 & 81-52 & 1 & {\bf M1-2}&  57-40 * (33-50 $\ddag$) &
 $\{0,\pm 1,\sqrt{2},\pm 3\}$ & \cite{pavicic-quantum-23} \\
 &  169-120 & 3 & {\bf M1-2}& 69-50 * (33-50 $\ddag$, 8-8 $\ddag$)  &
 $\!\!(0,\pm\omega,2\omega,\pm\omega^2,2\omega^2\}$ &
  \cite{pavicic-quantum-23}, see text\\
 &  597-358 & 3 & {\bf M1-2}& 95-66 * (58-66 $\ddag$)  &
  golden ratio---see text & see text \\
    \hline\end{tabular}}
\caption{3-dim NBMMPHs  
  obtained by methods {\bf M1-2} and {\bf M4-5}; 
  `$^*$' indicates that the MMPH is a KS; 
  `$\dag$' indicates that the MMPH is a noncontextual BMMPH;
  `$\ddag$' indicates that the MMPH is a non-KS contextual NBMMPH.}
\end{table}

\medskip

\subsubsection{\label{subsubsec:4D}  4-dim {\rm MMPH}s}

The 4-dim MMPHs are the most explored MMPHs in the literature
over the last three~decades.

In the beginning, several scientists spent five years seeking to
find three KS MMPHs with $\{0,\pm1\}$ components in three papers
~\cite{peres,kern,cabell-est-96a}. Fifteen years later, we are able
to generate all 1233 KS MMPHs from these components in a few
minutes on a PC from scratch \cite{pavicic-paris-video-2019}.

Later, large master sets based on serendipitous or intuitively
found connections of KS hypergraphs with polytopes or Pauli 
operators were found
~\cite{aravind10,waeg-aravind-jpa-11,waeg-aravind-fp-14,waeg-aravind-jpa-15,pavicic-pra-17}.
Unfortunately, they only partly generated 4-dim KS sets, and their
generation by means of these methods was neither automated nor
generalized. Instead, this was achieved by an automated generation
of MMPHs from basic components in~\cite{pm-entropy18,pwma-19}. 
Altogether, millions, if not billions, of 4-dim were generated in
the last 15 years, so we only include several new MMPHs in
Figure~\ref{fig:4Da} and refer to all of them in Table~\ref{T:4D}. 

The reason that we return to 4-dim MMPHs is that, originally,
we were only interested in their structure and did not give the
coordinatizations for them in our papers. Later on, we found a way
to generate MMPHs directly from vector components and found several
applications that required specifications of states/vectors. So,
here, we give coordinatizations for several chosen MMPH masters
generated from the simplest real and complex vector components,
what has not been done previously. We also find particular
symmetries in such generations. For~instance, $\{0,\pm 1,x\}$
and $\{0,\pm x,1/x\}$ might generate the same MMPH string,
although, of~course, not the same coordinatizations.

In Appendix \ref{app:2a}, we give strings and coordinatizations
for MMPHs generated from the following vector components:
$\{0,\pm 1,\phi\}$ and $\{0,\pm\phi,\frac{1}{\phi}\}$ yield
60-72, $\{0,\pm 1,i\}$ and $\{0,\pm i,1\}$ yield 86-152,
and $\{0,\pm i,\pm 1\}$ yield 92-185. The string and
coordinatization of 86-152, as well as the distribution of
the critical NBMMPHs contained in it, are given in
\mbox{Appendix \ref{app:2a}}.

\begin{figure}[H]
\vspace{-5pt}
  \includegraphics[width=\textwidth]{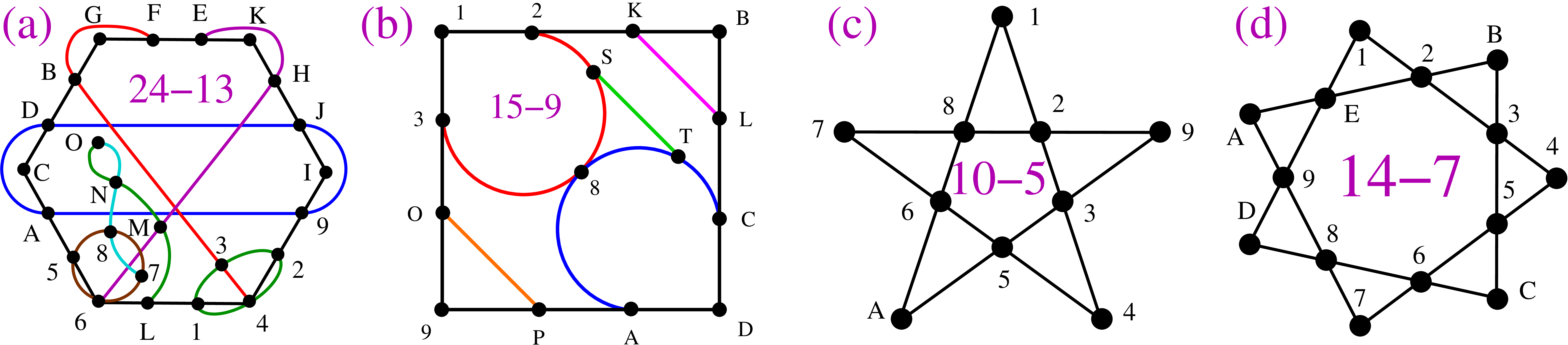}
\caption{\baselineskip=9pt 4-dim 
  NBMMPHs. (\textbf{a}) A 24-13 KS MMPH critical that does not
  have a $\{0,\pm1\}$ and is therefore not isomorphic to the 24-13
  BMMPH from the aforementioned 4-dim 1233 class; its
  coordinatization is given in Appendix \ref{app:2}; (\textbf{b}) the only
  critical non-KS NBMMPH obtained from 29-16 given in Appendix
  \ref{app:2}; its string and coordinatization are given in Appendix
  \ref{app:2}; (\textbf{c}) pentagram---see text; (\textbf{d}) heptagram---see text.}
\label{fig:4Da}
\vspace{-10pt}
\end{figure}

In fact, any new set of vector components might generate new
master sets and classes of MMPHs. This is apparently misunderstood
in~\cite{elford-lisonek-20}, which after finding that ``the vast
majority of known examples have been found by computer search\dots
without much insight in the sets generated'', 
offers a computer-free construction of an {\em infinite} family
of KS sets in a 4-dim space instead. However, once the sets are
generated, humans can offer any insight to them, and, once the
vector components are defined (\cite{elford-lisonek-20},
Equation~(1)), humans can generate MMPHs either by hand or by
a computer. A computer might be much faster.

In the end, we consider star-like MMPHs: the 4-dim regular pentagram
and heptagram (Schl\"afli symbols \{5/2\} and \{7/2\}) shown in
Figure~\ref{fig:4Da}(c,d). They are both KS MMPHs but apparently
without coordinatizations, meaning that our programs written in
C were run for days on a supercomputer, checking numerous vector
components, and did not yield any results. We also used
Mathematica to calculate the corresponding nonlinear equations
on a supercomputer; after a month, it faced the 2 TB memory
limitation. Note that the diagram in (\cite{hofmann-24}, Figure~1)
(which is graphically and misleadingly tantamount to
Figure~\ref{fig:4Da}(c)) is not an MMPH (e.g., D1 and D2
are not orthogonal) and that, therefore, the states
assigned to the vertices in that diagram do not
represent a coordinatization of the MMPH pentagram.
Moreover, the~diagram $\{{\tt 1,2,S2,f,S1,1}\}$, which
is called a pentagram in~\cite{hofmann-24}, is not
a Schl\"afli \{5/2\} pentagram. 

\begin{table}[H]
\vspace{-5pt}
\setlength{\tabcolsep}{3.3pt}
{\begin{tabular}{|c|cccccc|}
    \hline
  \multirow{3}{*}{\textbf{dim}} & \multirow{3}{*}{\textbf{Master}} &
  \textbf{No.~of}  & \multirow{3}{*}{\textbf{Methods}} 
  & \multirow{2}*{\textbf{KS and Non-KS}}
  & \multirow{2}*{\textbf{Vector}} & \multirow{3}{*}{\textbf{References}} \\
&    &  \textbf{Non-Isom} & & \multirow{2}*{\textbf{Criticals}} &
      \multirow{2}*{\textbf{Components}} & \\
&       & \textbf{Criticals}    &  &     & & \\
  \hline
\multirow{6}{*}{4-dim} & 86-152 & $>$8 millions &{\bf M1-6}& 10-7 \ddag, 15-9 \ddag, 18-9, 24-13 & 
   $\{0,\pm 1,i\}$ & see text \\ 
& 92-185 & 600,000 &{\bf M1-6}& 18-9 (smallest) & 
$\{0,\pm 1,\pm i\}$ & \cite{pavicic-pra-17,pwma-19}  \\
 &888-1080&$>$1.5 billions  & {\bf M1-6} & 18-9 (smallest) & 
 $\{0,\pm\phi,\frac{1}{\phi}\}$ & \cite{pavicic-pra-17,pwma-19}
  \\  
 &400-1012& $>$250,000 & {\bf M1-6} & 18-9 (smallest) & 
 $\{0,\pm 1,\pm\omega,\pm\omega^2\}$ & \cite{pavicic-pra-17,pwma-19}
 \&\ here \\  
 &10-5& 1 & {\bf M6} & 10-5 & 
   none ? & see text\\  
 &14-7& 1 & {\bf M6} & 14-7 & 
   none ? & see text\\  
   \hline\end{tabular}}
\vspace{-5pt}
\caption{\baselineskip=9pt 4-dim  NBMMPHs obtained as described
  in the text;
  `$\ddag$' indicates that the MMPH is a non-KS contextual NBMMPH.
  The first critical is given in Appendix~\ref{app:2}; others are
  given in Figure~\ref{fig:4Da} and Appendix~\ref{app:2};
  ``none ?'' means that there might not be any coordinatization.}
\label{T:4D}
\vspace{-15pt}
\end{table}

\subsubsection{\label{subsubsec:5-8D}  5- to 8-dim {\rm MMPHs}}

In higher dimensions, the simplest vector components $\{0,\pm 1\}$
allow us to obtain the most straightforward generation of MMPHs, although,
in, e.g.,~the 6-dim space, the~vector components $\{0,1,\omega\}$
generate smaller MMPHs than $\{0,\pm 1\}$. Thus, we shall give
examples of known vector components whenever available and/or
interesting.

Moreover, in~contrast to our previous generations, where we made
use of mostly {\bf M5}, here, in~order to generate non-KS MMPHs,
we chiefly use {\bf M3}. In~other words, to~obtain non-KS MMPHs,
we do not remove vertices with $m>1$ from KS MMPHs, as~we mainly
did in~\cite{pavicic-entropy-23}, but~only those with $m=1$. 

In the 5-dim space,  $\{0,\pm 1\}$ components generate the 105-136
master, whose distribution we presented in~\cite{pavicic-pra-22}.
There, we gave one of the two smallest KS MMPHs; in
Figure~\ref{fig:5-6d}(a), we give the other (the string and
coordinatization are given in Appendix \ref{app:3}). It does not
have vertices with $m=1$, and, overall, the class has a comparatively
small number of NBMMPHs that contain $m=1$ vertices.~Moreover, when they
do, only a few of them generate small non-KS MMPHs. One such
MMPH (11-7) is shown in Figure~\ref{fig:5-6d}(b). It is obtained from
one of 38-20 KS MMPHs, but tracing the coordinatization down to
the filled 11-7 and then to the 11-7 itself requires an algorithm
whose calculation is more time-consuming than filling 11-7 up
to 21-7 by MMPSHUFFLE and determining the coordinatization by VECFIND
within less than 1 sec on a PC. The~strings and coordinatization
are given in Appendix \ref{app:3}.

In the 5-dim space, a master MMPH obtained from $\{0,1,\omega\}$
vector components is a noncontextual BMMPH, in contrast to the 6-dim
space, where these three components generate a star-like KS NBMMPH,
which cannot be generated from $\{0,\pm 1\}$ components. (Star-like
MMPHs exist only in even-dimensional spaces; we shall
return to star-like constructions in Section~\ref{subsec:app}.)

In the 6-dim space, vector components $\{0,1,\omega\}$ 
generate a 216-153 master, which contains over 15 million
non-isomorphic KS MMPHs but just three critical KS MMPHs,
21-7, 27-9, and~33-11 (3$l$-$l$, odd $l>5$), the~first
two of which are shown in (\cite{pm-entropy18}, Figure~A2). The~third
one, 33-11, which has not been previously presented in the star-like
form, is given in Figure~\ref{fig:5-6d}(c), with its string and
coordinatization in \ref{app:4}. We can continue the 3$l$-$l$
construction (see, e.g.,~the 39-13 in  (\cite{pm-entropy18}, Figure~A2)),
but then we have to enlarge the set of vector component generators to
$\{0,1,\omega,\omega^2\}$ (ibid.).

\bigskip
\begin{figure}[h]
  \includegraphics[width=\textwidth]{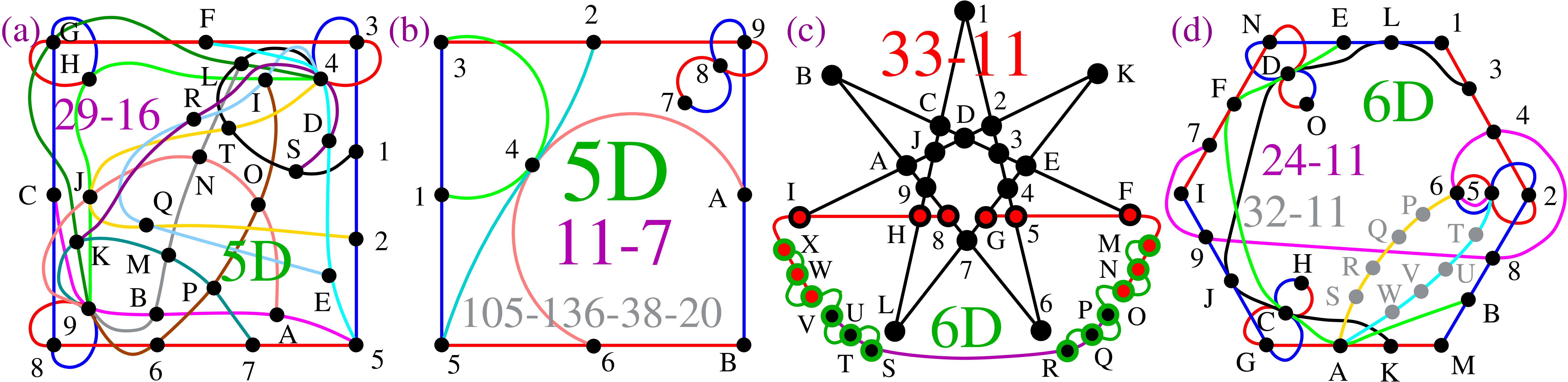}
  \caption{(\textbf{a}) One of the two smallest 5-dim criticals from the 105-136
  KS master (the other was shown in (\cite{pavicic-pra-22}, Figure~2b);
  its string and coordinatization are given in Appendix \ref{app:3};
  (\textbf{b}) the smallest critical non-KS NBMMPH obtained from one of over
  3000 38-20 KS MMPHs from the 105-136 class; the filled 11-7 is
  not necessarily a subhypergraph of the 38-20; their strings and
  coordinatization are given in Appendix \ref{app:3}; (\textbf{c},\textbf{d}) see text.}
\label{fig:5-6d}
\end{figure}

The 6-dim $\{0,\pm 1\}$ vector components generate the 236-1216 KS
master, which contains more than 3.7 million critical KS MMPHs
(\cite{pavicic-pra-17}, Figure~12). A number of graphical
representations of small critical MMPHs are given in
(\cite{pavicic-pra-17}, Figure~11). Another example---32-11---is
given in Figure~\ref{fig:5-6d}(d). It contains a non-KS NBMMPH

24-11. Their strings and coordinatization are given in Appendix
\ref{app:4}.

The 7- and 8-dim KS and non-KS NBMMPHs generated by $\{0,\pm 1\}$
vector components are extensively elaborated on in
~\cite{pavicic-pra-17,pwma-19,pavicic-pra-22,pw-23a,pavicic-entropy-23}.
See Table~\ref{T:5-8}. 

\begin{table}[H]
  \setlength{\tabcolsep}{5pt}
  {\begin{tabular}{|c|cccccc|}
    \hline
  \multirow{3}{*}{\textbf{dim}} & &
  \textbf{No.~of } & \multirow{3}{*}{\textbf{Methods}} 
  & \multirow{2}*{\textbf{Small}}
  & \multirow{2}*{\textbf{Vector}} & \multirow{3}{*}{\textbf{References}} \\
&  \textbf{Master}   &  \textbf{Non-Isom} & & \multirow{2}*{\textbf{MMPHs}} &
      \multirow{2}*{\textbf{Components}} & \\
&        & \textbf{Criticals}    &  &     & & \\
  \hline
\multirow{2}{*}{5-dim} &105-136& $>$27.8 millions & {\bf M1}& 29-16 * & 
  $\{0,\pm 1\}$ & see text\\
 &38-20& - & {\bf M2-3} & 5-11 \ddag & 
   $\{0,\pm 1\}$ & see text\\
  \hline
\multirow{2}{*}{6-dim} &216-153& 3 & {\bf M1}& 33-11 * & 
  $\{0,1,\omega\}$ & see text\\
  &236-1216& $>$3.7 millions & {\bf M2-3}& 32-11\dag~(24-11 \ddag) & 
   $\{0,\pm 1\}$ & see text\\
  \hline
\multirow{2}{*}{7-dim}
  &47-176& $>$1 million &{\bf M2,8}&34-14 *& $\{0,\pm 1\}$ &
  \cite{pavicic-pra-22,pw-23a}\\
  &805-9936& $>$42,800 &{\bf M1,3}&14-18 \ddag& $\{0,\pm 1\}$ &
  \cite{pavicic-entropy-23}\\
  \hline
  8-dim
  &3280-1361376& $>$7 millions &{\bf M1-2}&36-9 * (15-9 \ddag)&
  $\{0,\pm 1\}$ &
  \cite{pavicic-pra-22,pavicic-entropy-23}\\
   \hline\end{tabular}}
\caption{NBMMPHs obtained by methods {\bf M1-3} and {\bf M8};
  `$^*$' indicates that the MMPH is a KS
  `$\ddag$' indicates that the MMPH is a non-KS contextual NBMMPH.}
\label{T:5-8}
\end{table}
\unskip

\subsubsection{\label{subsubsec:9-32D}  9- to 32-dim {\rm MMPHs}}

Since the generation of MMPH masters from small vector components is
an exponentially complex task, it is unfeasible to generate
them in dimensions greater than eight. In~\cite{pavicic-pra-22}, we
successfully obtained a 9-dim master 9586-12068704 from the
$\{0,\pm1\}$ vector components. However, this master contains such
a high number of BMMPHs that even a month-long run on a supercomputer
did not yield any KS NBMMPHs. We were only able to extract non-KS
NBMMPHs, since they are significantly more~abundant.

Fortunately, our dimensional upscaling method {\bf M8} does not
scale in complexity with increasing dimensions. This allows us to
generate MMPHS using a bottom-up approach, rather than relying on
top-down generation from masters created from small vector
components. In~\cite{pw-23a}, we present the distributions of MMPHs
obtained via this method across dimensions 9 to 16 and 27, as~well as
small instances of them. There, we also give a table of KS MMPHs for
all these dimensions that are analogous to those in the tables above.
In~\cite{pavicic-entropy-23}, we focus on generating non-KS MMPHs in
dimensions 9 to 16. In~\cite{pavicic-pra-22}, we generate 32-dim KS
MMPHs from a five-qubit set derived in~\cite{planat-saniga-12}. 

In Figure~\ref{fig:9-15}, we offer a new graphical~representation
of some non-KS MMPHs from~\cite{pavicic-entropy-23},
{\parindent=0pt providing a more comprehensive insight into their structure.
Specifically, rather than depicting an MMPH by the largest loop
formed by its hyperedges, we use circles or parts of circles to
illustrate the hyperedges. The reader can compare
Figure~\ref{fig:9-15}(a--e) with 
(\cite{pavicic-entropy-23}, Figures~4b,c and 5a--c) to observe the
difference. This presentation method allows us to visualize the
so-called $\delta$-feature~\cite{pavicic-pra-17} (in an $n$-dim
space, hyperedges might share up to $n-2$ vertices;
cf.~Definition~\ref{def:MMP-string}(4)) as overlapping semi-circles
(cf.~Figure~\ref{fig:9-15}(c--e)).}

Regarding higher-dimensional examples, Figure~\ref{fig:27-32}(a)
presents a 27-dim non-KS MMPH\linebreak
$\overline{\rm subhypergraph}$ 36-5
obtained via {\bf M2,3} from the 27-dim KS MMPH 141-16 master
generated through method {\bf M8} \cite{pw-23a}. Consequently, there
are gaps in the vertex numbering. While we could have closed these
gaps, we chose not to do so in order to allow the reader to derive
the coordinatization and reconstruct the missing vertices from the
master 141-16 provided in~\cite{pw-23a}. As~for the ``missing''
vertices---specifically dropped vertices with $m=1$---the circular
hyperedge still contains $m=1$ (grey) vertices based on our
principle of ensuring that all $n$ vertices are included in at
least one hyperedge. This principle enables the algorithms and
programs to recognize a particular MMPH as belonging to an $n$-dim
space. Furthermore, the~automated reduction of 141-16 to the critical
36-5 ensures that the latter MMPH would remain an NBMMPH, meaning
that it is still contextual, even if $m=1$ vertices from the circular
hyperedge are removed. However, if~we added one or more vertices to,
for~example, the pink hyperedge {\tt 3CL}, the~MMPH would cease to be
contextual and would instead become a BMMPH. For~further discussion
of this property, see
(\cite{pavicic-pra-22}, Supplementary~Material~p.~6).

\begin{figure}[t]
  \includegraphics[width=\textwidth]{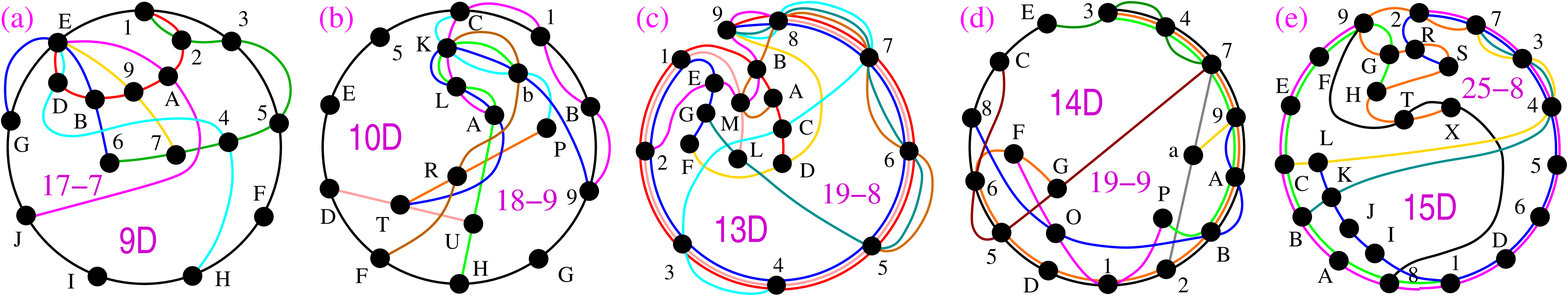}
\caption{(\textbf{a}) A 9-dim non-KS 17-7 MMPH obtained from 19-8 in
  (\cite{pavicic-entropy-23}, Figure~4b) by dropping $m=1$ vertices
  {\tt C} and {\tt L} and hyperedge {\tt EFI}; note that the 19-8
  stops being critical when {\tt L} is removed, although it is an
  $m=1$ vertex; (\textbf{b}) a 10-dim non-KS 18-9 MMPH is presented via a
  circle of the largest hyperedges; (\textbf{c}--\textbf{e}) non-KS MMPHs presented
  via overlapping semi-circles (parts of hyperedges) featuring
  the $\delta$-property---see text; the strings and coordinatizations
  of the MMPHs are given in~\cite{pavicic-entropy-23}.}
\label{fig:9-15}
\end{figure}

In Figure~\ref{fig:27-32}(c), we give a 32-dim non-KS MMPH 40-5
generated by {\bf M2,3,5} from the 32-dim KS MMPH 144-11 obtained
in~\cite{pavicic-pra-17}. All points that hold for the 27-dim
$\overline{\rm subhypergraph}$ above hold for this one as well. 
For instance, if~we added a vertex to the black hyperedge
{\tt dfhk}, the MMPH would lose its~contextuality. 

\begin{figure}[H]
  \includegraphics[width=\textwidth]{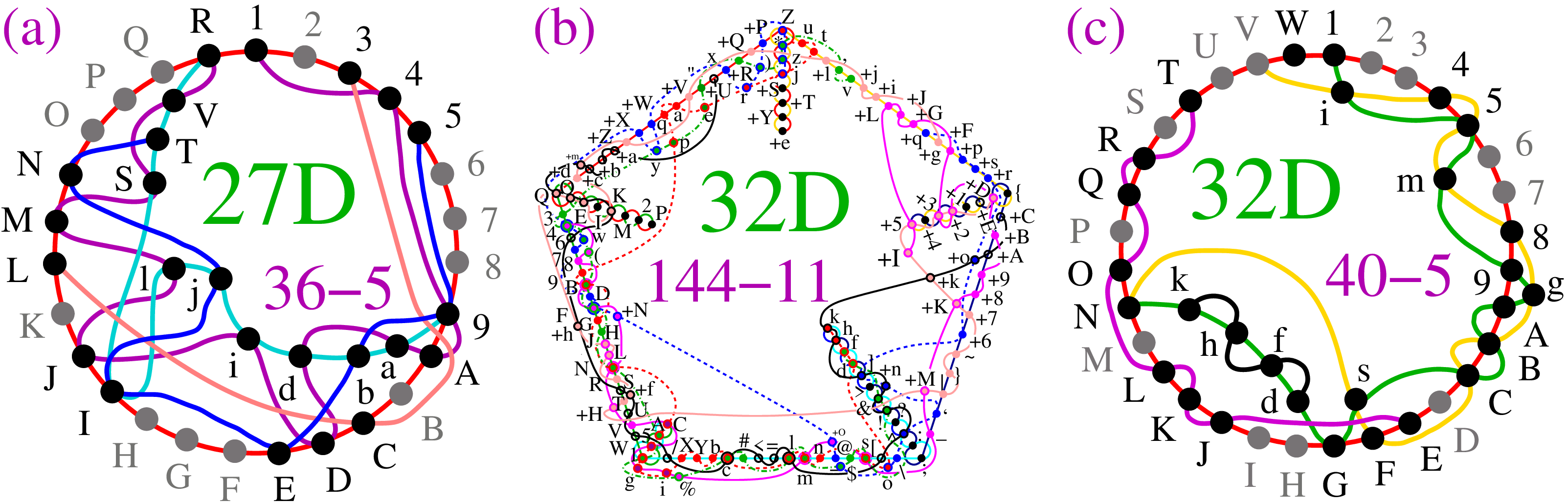}
  \caption{(\textbf{a}) A small 27-dim critical non-KS NBMMPHs
  obtained in this paper; its coordinatization can be derived from
  its master's (141-16) vertices (\cite{pw-23a}, Appendix~14);
  (\textbf{b}) a 32 KS NBMMPH obtained in~\cite{pavicic-pra-17};
  (\textbf{c}) a critical non-KS NBMMPH $\overline{\rm subhypergraph}$
  obtained from it in this paper; see text for details.}
\label{fig:27-32}
\end{figure}

The list of non-KS NBMMPHs considered in this section is given in
Table~\ref{T:small2}. 

\begin{table}[H]
\center
\setlength{\tabcolsep}{10pt}{\begin{tabular}{|c|ccc|}
    \hline
  {\textbf{dim}} & \textbf{Smallest Critical Non-KS NBMMPHs} &
  \textbf{Master}  & \textbf{Vector Components}\\
  \hline
   {9-dim} &17-7& 47-16 & $\{0,\pm 1\}$ \\
   {10-dim}&18-9& 50-15 & 
   $\{0,\pm 1\}$ \\
   {11-dim}&19-8& 50-14 & 
   $\{0,\pm 1\}$ \\
   {12-dim}&19-9 & 52-9 & 
   $\{0,\pm 1\}$ \\
   {13-dim}&19-8 & 63-16 & 
   $\{0,\pm 1\}$ \\
   {14-dim}&19-9 & 66-15 & 
   $\{0,\pm 1\}$ \\
   {15-dim}&25-8 & 66-14 & 
   $\{0,\pm 1\}$ \\
   {16-dim}&22-9 & 70-9 & 
   $\{0,\pm 1\}$ \\
   {27-dim}&36-5 & 141-16 & 
   $\{0,\pm 1\}$ \\
   {32-dim}&40-5 & 144-11 & 
   $\{0,\pm 1\}$ \\
  \hline
\end{tabular}}
\caption{The smallest critical non-KS MMPHs obtained 
    via {\bf M2}, {\bf M7}, and~{\bf M8}. Notice the steady
    fluctuation in the number of hyperedges---the minimum
    complexity of NBMMPHs does not grow with the dimension.}
  \label{T:small2}
\end{table}

\subsection{\label{subsec:app}  Applications}

Possible applications of contextual sets/hypergraphs in
higher dimensions do not face challenges in the generation of
MMPHs but rather in their implementation. Apparently, the most
feasible implementations would be those utilizing the angular
momentum of photons in a holographic approach.
Nevertheless, in~this section, we will consider applications
using a theoretical approach, addressing realistic
implementation limitations only as~necessary.

\subsubsection{\label{appA}Larger~Alphabet}

We extend the larger alphabet procedure discussed in~\cite{bech-00}.
A 4D KS ``protection'' for quantum key distribution (QKD) protocols
was proposed in~\cite{cabello-dambrosio-11} based on a modification
of the BB84 protocol outlined in~\cite{svozil-10}. A~KS hypergraph
with nine edges has been employed in this context. The~protocol runs
as follows: (i) Alice randomly picks one of the nine hyperedges
(bases) and sends Bob a randomly chosen state (vertex) from that
hyperedge; (ii) Bob randomly picks one of the nine hyperedges and
measures the system received from Alice. Instead of qubits, we are
working with ququarts, allowing the transfer of not 1 but 2 bits of
information~\cite{bech-00}. We can modify and generalize this QKD
protocol to apply it to any $k$-$l$ hypergraph ($k$ vertices,
$l$ edges). However, this does not provide a quantum advantage
to Alice and Bob. The~reason is as~follows.

\begin{wrapfigure}{r}{0.25\textwidth}
  \vspace{-25pt}
  \begin{center}
  \includegraphics[width=0.24\textwidth]{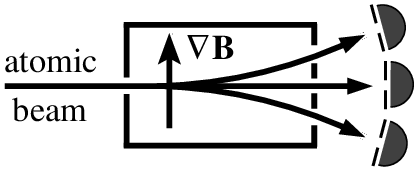}
  \end{center}
\vspace{-15pt}
  \caption{\baselineskip=9pt Stern--Gerlach experiment with
  spin-one atoms (\cite{feynmanIII}, Figure~5-1).}
\label{fig:st-ger}
\vspace{-15pt}
\end{wrapfigure}
Both KS and non-KS contextual sets (NBMMPHs) yield measurement
outputs that differ from the predetermined measurement outputs
that we would expect from classical sets that adhere to rules
{\em(i)} and {\em(ii)} outlined in Definition~\ref{def:n-b}.
Let us illustrate this by means of the following simple
Stern--Gerlach experiment. Quantum measurements, as illustrated  
in Figure~\ref{fig:st-ger}, ideally always\break
trigger one of the detectors positioned at each
of the output ports,
irrespective of whether the Stern--Gerlach devices are considered
individually or networked (joined together) within an MMPH structure
(see Quantum Indeterminacy Postulate \ref{postulate}). This contrasts
with the classical counterparts. Individually, the~classical devices
produce the same outputs as the quantum ones, but,~when networked
and assumed to have predetermined output values, they should yield
different results for at least one of the measurements of the
devices/gates/hyperedges, meaning that, for this measurement, no
detector should be triggered, i.e.,~all three vertices/outputs
should be assigned a value of 0. Consequently, such networked
classical devices are~not feasible.

In other words, Alice and Bob can only achieve a quantum advantage
if their eavesdropper, Eve, assumes that they communicate using
outputs from networked classical devices. However, Eve will not
assume this, since she knows that it is impossible due to the KS
theorem. Instead, Eve will straightforwardly introduce fake
messages for every hyperedge/gate, thus mimicking the quantum
outputs and disregarding rules {\em(i)} and {\em(ii)} of
Definition~\ref{def:n-b}. Therefore, the ``hybrid ququart-encoded
quantum cryptography protected by Kochen--Specker contextuality''
\cite{cabello-dambrosio-11} is not truly ``protected'' against
any Eve; it is simply another version of the BB84 protocol that
offers no quantum~advantage.

However, Alice and Bob can implement the following protocol
in which Bob sends messages to Alice. We can assume that it is
carried by photons carrying orbital angular momentum; over
$100\times 100$ entangled angular momentum dimensionality has
been achieved experimentally~\cite{krenn-zeilinger-14}.
It can also be implemented with higher-spin Stern-Gerlach
devices \cite{tekin-16}.\linebreak

\begin{itemize}
\item Alice picks up $n$-dim MMPHs and sends outputs from the
  gates/hyperedges of the chosen MMPHs to Bob in blocks; she
  can repeat sending from the same MMPHs or pick up new ones;
\item Bob stores Alice's sending in quantum memory;
\item Alice informs Bob about which sending belonged to which
  hyperedge and from which MMPH over a classical channel, with~a
  delay; 
\item Bob reads off each Alice's sendings and sends them back to
  her, scrambled, over the quantum channel; scrambling codes
  transform Bob's sendings into his messages but they are still
  undisclosed to Alice;  
\item Alice stores Bob's sendings in quantum memory;
\item Bob informs Alice of the scrambling code over a classical
  channel, with~a delay; 
\item After an agreed number of exchanged blocks, they can transmit
  some messages over a classical channel to check whether Eve is in
  the quantum channel;  
\item After Alice has correlated the reflected sendings with the
  original ones with the help of Bob's code, she learns how to
  measure each of them from the quantum memory and read off Bob's
  message.
\end{itemize}
  
All steps in the protocol are completely automated.
The probability that Eve might introduce correct sendings in
the channel when Alice and Bob transfer the output of, say, the~32-dim
MMPH with 11 hyperedges, shown in Figure~\ref{fig:27-32}(b), is less
than $3\times 10^{-17}$. Therefore, privacy amplification is
hardly~necessary.

\begin{wrapfigure}{r}{0.30\textwidth}
\vspace{-10pt}
  \begin{center}
    \includegraphics[width=0.30\textwidth]{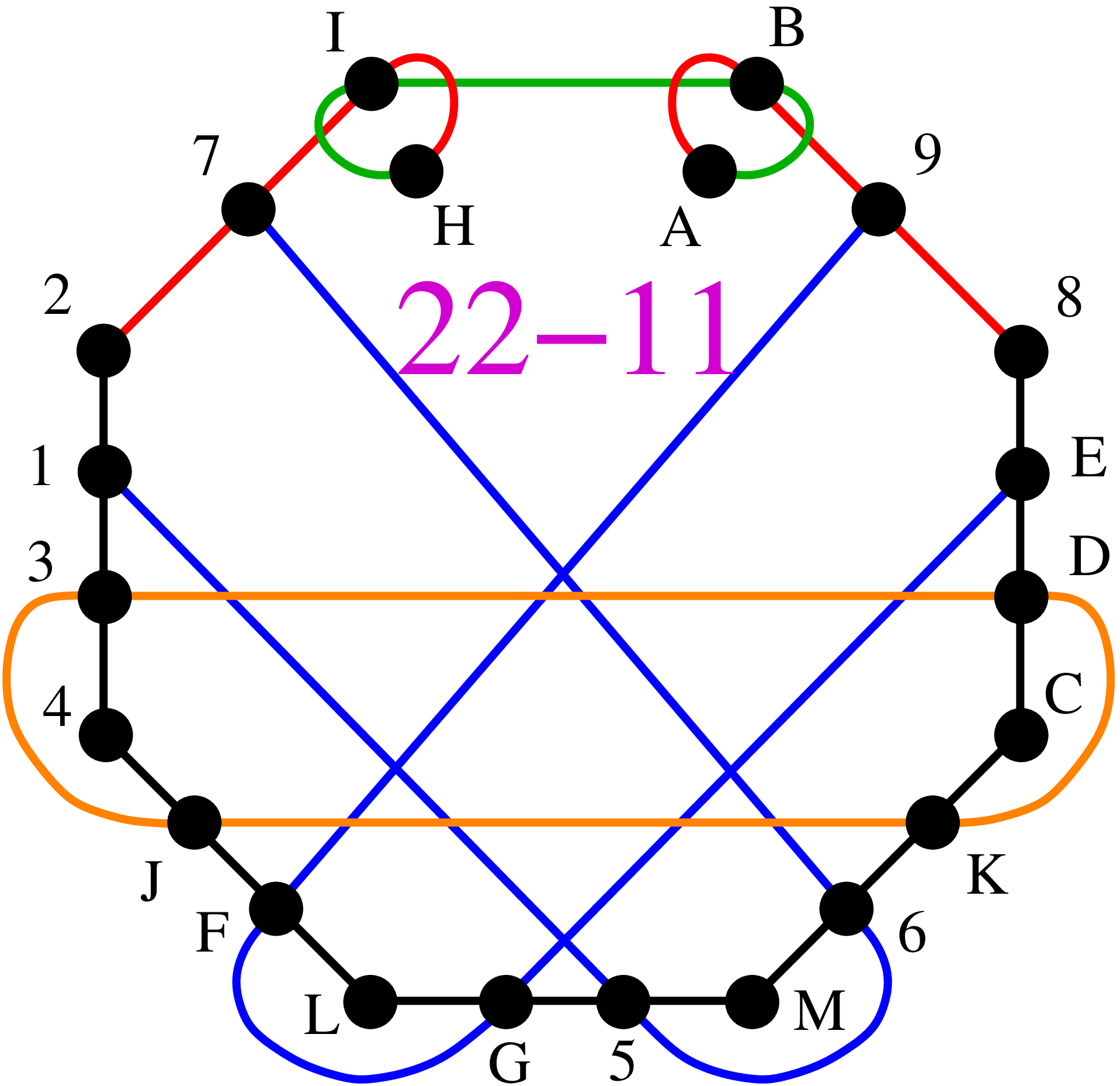}
  \end{center}
\vspace{-15pt}
  \caption{\baselineskip=9pt The 22-11 obtained from
    the 95-185 KS MPPH generated by vector components
    $\{0,\pm 1,\pm i\}$ or $\{0,\pm 1,i\}$.}
  \label{fig:22-11}
\vspace{-15pt}
\end{wrapfigure}

The advantage of the higher-dimensional protocol over the BB84 can
be well exemplified by the fairly simple 22-11 KS
MMPH shown in Figure~\ref{fig:22-11}. 

The 22-11 cannot have a coordinatization based on $\{0,\pm1\}$
vector components, but~it can have this based on $\{0,\pm1,i\}$, $\{0,i,\pm 1\}$,~$\{0,\pm1,\pm i\}$, or~$\{0,\pm1,2\}$, etc.
Below are the string and coordinatization generated by the first
two sets of vector components via the 86-152 master MMPHs
(whose string, coordination, and~distribution are given
in \mbox{Appendix \ref{app:2a}}).

{\bf 22-11}\qquad
  {\tt 4312,27IH,HIBA,AB98,8EDC,CK6M,M5GL,LFJ4,\linebreak 1567,9FGE,J3DK.}
  {\parindent=0pt  {\tt 1} = (0,0,0,1), {\tt 8} = (0,0,1,1), {\tt 9} = (0,0,1,$-$1),}

    {\parindent=0pt {\tt H} = (0,0,1,i), {\tt I} = (0,0,i,1), {\tt 2} = (0,1,0,0), {\tt J} = (0,i,0,1),}

    {\parindent=0pt {\tt 5} = (0,1,i,0), {\tt 6} = (0,i,1,0), {\tt 7} = (1,0,0,0), {\tt M} = (1,0,0,i),}

    {\parindent=0pt {\tt 3} = (1,0,i,0),{\tt 4} = (i,0,1,0), {\tt F} = (1,1,i,i), {\tt C} = (1,1,i,$-$i),}

    {\parindent=0pt {\tt D} = (1,1,$-$i,i), {\tt G} = (i,i,1,1), {\tt E} = (1,$-$1,0,0), {\tt L} = (1,$-$1,i,$-$i),}

    {\parindent=0pt {\tt K} = ($-$1,1,i,i), {\tt A} = (1,i,0,0), {\tt B} = (i,1,0,0)}

The coordinatization enables implementations via two
qubits mounted on single photons by means of linear and
circular polarization and orbital angular momentum. To~see
this, let us first define the photon qubit states:
\begin{eqnarray}
 &&|H\rangle=
\begin{pmatrix}
      1\\
      0\\
    \end{pmatrix}_{1},\qquad  
|V\rangle=
\begin{pmatrix}
      0\\
      1\\
\end{pmatrix}_{1},\quad  
|D\rangle=\frac{1}{\sqrt{2}}\left(
     \begin{matrix}
      +1\\
      1\\
    \end{matrix}
  \right)_{1},\quad
|A\rangle=\frac{1}{\sqrt{2}}\left(
     \begin{matrix}
      -1\\
      1\\
    \end{matrix}
  \right)_{1},\quad  
|R\rangle=\frac{1}{\sqrt{2}}\left(
     \begin{matrix}
      1\\
     +i\\
    \end{matrix}
  \right)_{1}, \nonumber\\
&&|L\rangle=\frac{1}{\sqrt{2}}\left(
     \begin{matrix}
      1\\
     -i\\
    \end{matrix}
  \right)_{1},\quad
|2\rangle=\left(
     \begin{matrix}
      1\\
      0\\
    \end{matrix}
  \right)_{2},\quad 
|-2\rangle=\left(
     \begin{matrix}
      0\\
      1\\
    \end{matrix}
  \right)_{2},\quad 
  |h\rangle=\frac{1}{\sqrt{2}}\left(
     \begin{matrix}
      1\\
      1\\
    \end{matrix}
  \right)_{\!2},\quad
  |v\rangle=\frac{1}{\sqrt{2}}\left(
     \begin{matrix}
      1\\
     -1\\
    \end{matrix}
  \right)_{\!2},
  \label{eq:21-11}
\end{eqnarray}
where $H,V$ are horizontal and vertical, $D,A$ are diagonal and
anti-diagonal, and~$R,L$ are right and left circular polarizations,
while $\pm 2$ are Laguerre--Gauss modes carrying
$\pm 2\hbar$ units of orbital angular momentum (OAM) and
$h,v$ are their $\pm$ superpositions, respectively. Indices `1'
and `2' refer to the first and second qubits mounted on the system
(single photon), respectively.

These qubit states build the hyperedge gates as tensor
products, e.g.,
\begin{eqnarray}
&&|{\tt G}\rangle=\frac{1}{\sqrt{2}}
  \begin{pmatrix}
    i\\
    i\\
    1\\
    1\\
  \end{pmatrix}
  = \frac{1}{2}
  \begin{pmatrix}
    i  \begin{pmatrix}
      1\\
      1\\
    \end{pmatrix}\\
    1  \begin{pmatrix}
      1\\
      1\\
    \end{pmatrix}\\
      \end{pmatrix}  
  = \frac{1}{\sqrt{2}}
  \begin{pmatrix}
      i\\
      1\\
    \end{pmatrix}_{1}\otimes
  \frac{1}{\sqrt{2}}\begin{pmatrix}
      1\\
      1\\
    \end{pmatrix}_{2}
  =-i|L\rangle_1|h\rangle_2,\nonumber\\
&&|{\tt L}\rangle=\frac{1}{\sqrt{2}}
  \begin{pmatrix}
    1\\
    -1\\
    i\\
    -i\\
  \end{pmatrix}
  = \frac{1}{2}
  \begin{pmatrix}
    1  \begin{pmatrix}
      1\\
      -1\\
    \end{pmatrix}\\
    i  \begin{pmatrix}
      1\\
      -1\\
    \end{pmatrix}\\
      \end{pmatrix}  
  = \frac{1}{\sqrt{2}}
  \begin{pmatrix}
      1\\
      i\\
    \end{pmatrix}_{1}\otimes
  \frac{1}{\sqrt{2}}\begin{pmatrix}
      1\\
      -1\\
    \end{pmatrix}_{2}
  =|R\rangle_1|v\rangle_2
  \label{eq:22-11-2}
\end{eqnarray}

Some states (e.g., {\tt 5}) require rather involved manipulation
(including superpositions of tensor products).

Fortunately, in a large alphabet quantum cryptography, we can
limit ourselves to $\{0,\pm 1\}$ vector components in any
higher dimension. It is only in the 4-dim space that there
are no more than six critical MMPHs with a coordinatization
based on $\{0,\pm 1\}$ vector~components. 

Here, with~one of the smallest 4-dim MMPHs, the~states build
eleven hyperedges/gates, which cause Eve's probability
of correctly guessing all states to be less than $10^{-6}$. 

Notice the orthogonality of complex vectors, e.g.,
{\tt A}$\cdot{\tt B}^*=(1,i,0,0)\cdot(i,1,0,0)^*
=(1,i,0,0)\cdot(-i,1,0,0)=0$

Note that the generation of the states from Alice and Bob's
devices is genuinely random (quantum randomness)
\cite{stipcevic-16}. Moreover, Bob scrambles Alice's sending
through a quantum random number~generator.

One can object that the protocol could be implemented using
other quantum sets and not just MMPHs. While this is true,
the~abundance and automated generation of MMPHs of any size and
dimension make them the most favourable candidates for large
alphabet communication in higher~dimensions.

\subsubsection{\label{appB}Oblivious Communication~Protocol}

Communication in a system with no dimensional bound and with some
information about the sender's input unrevealed, i.e.,~oblivious
communication, is presented in~\cite{saha-hor-19}.

However, Alice and Bob can have a quantum advantage only
if their eavesdropper (Eve) assumes that they communicate using
outputs obtained from networked classical devices. Since they do
not, the~oblivious communication protocol needs to be modified so
as to follow the protocol proposed in Section~\ref{appA}. Then,
the~protocol could utilize any KS MMPH in any~dimension.

\subsubsection{\label{appC}Generalized Hadamard~Matrices}
\medskip
Most quantum computation algorithms are based on the Fourier
transform, of~which the Hadamard ($H$) transform is a special case.
Recently, the~$S$ class of $H$ matrices, known as $S$-$H$ matrices
in $\mathbb{C}^n$, with~$n$ being even, has been designed to prove
the existence of KS hypergraphs in an $n$-dim space~\cite{lisonek-19}.
Our method generates any of these KS hypergraphs, which are all found
to be star-like (see Figure~\ref{fig:stars}). Inversely, it allows us
to generate the elements of the corresponding $S$-$H$ matrices.
The~$S$-$H$ matrices depend on the following~theorem. 

\begin{figure}[H]
  \center
  \includegraphics[width=0.8\textwidth]{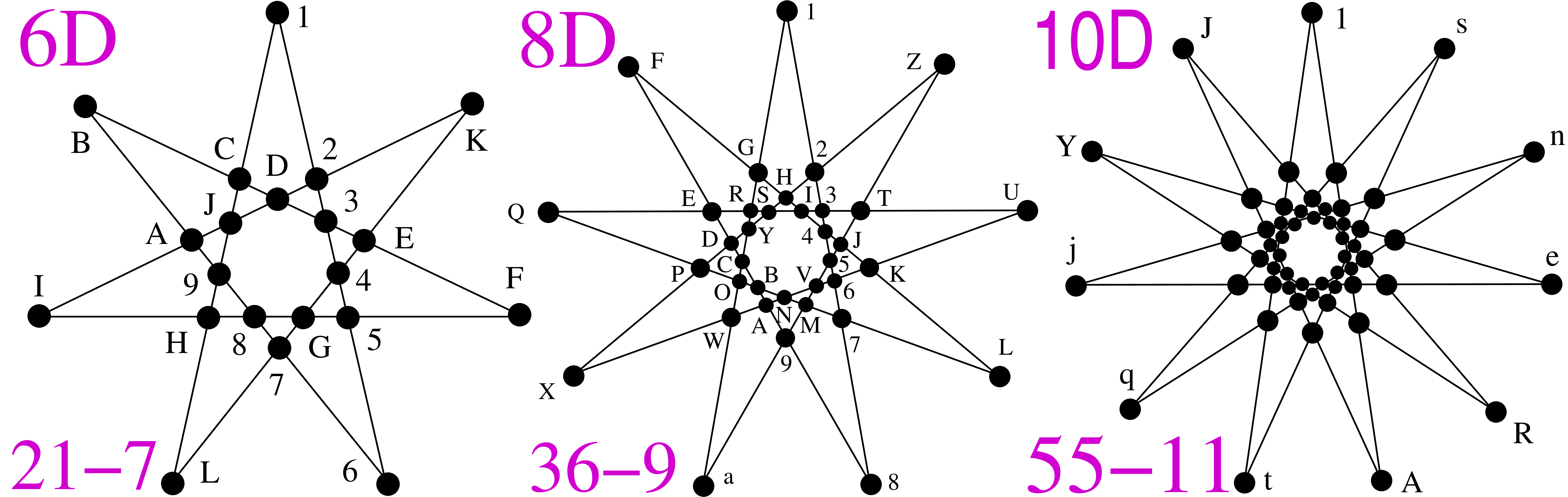}
\caption{Star-like
  $\frac{n}{2}(n+1)$-$(n+1)$,
  ($n$ even),
  $n$-dim MMPHs related to $S$-$H$ matrices; see text.}
\label{fig:stars}
\end{figure}

\begin{Xeorem}\label{th:Lisonek} (Lison{\v e}k 2019) Suppose that
  there exists an $S$-$H$ matrix of order $n$ ($n$ even); then,
  there exists a {\em KS} hypergraph $k$-$l$ in $\mathbb{C}^n$ such
  that $k=\binom{n+1}{2}$ and $l=n+1$.
\end{Xeorem}

In~\cite{lisonek-19}, Lison{\v e}k defines a KS hypergraph in
$\mathbb{C}^n$ (Definition~1). By~Definition~2.1, he defines an
$S$-$H$ matrix and, in Definition~2.2, a generalized Hadamard
matrix. Via \mbox{Theorem \ref{th:Lisonek}}, he connects a KS
hypergraph and a corresponding
$S$-$H$ matrix. The~proof provides an algorithm for a mutual
one-to-one mapping of their elements in any even dimension.
However, the~ details are beyond the scope of the present paper,
and we direct the reader to ref.~\cite{lisonek-19}. Definitions
1, 2.1, 2.2 and Theorem 3 above are from ref.~\cite{lisonek-19}.

Only the 6D $21$-$7$ KS hypergraph was known to Lison{\v e}k in
~\cite{lisonek-19}, i.e.,~only one particular $S$-$H$ matrix,
aside from the {\em existence} of all of them. Our method
generates any of these KS hypergraphs together with their
coordinatization (although the time required to generate the
coordinatization rises exponentially with the dimension, the time
required to generate MMPHs themselves does not). Inversely,
as is the core of this application, it gives the elements
of the corresponding $S$-$H$ matrices. It further demonstrates
that all the corresponding KS hypergraphs are star-like; see
Figure~\ref{fig:stars}. This feature clarifies why $n$ has to
be even: one cannot draw a regular star with the Schl\"afli
symbol \{$n$+1/$\frac{n}{2}$\} in odd-dimensional spaces
because $\frac{n}{2}$ has to be an integer. The~6,- 8-,
and~10-dim star strings and coordinatizations are given in
(\cite{pw-23a}, Appendix). 

\subsubsection{\label{appD}Stabilizer Operations}

Stabilizer operations entail MMPHs
~\cite{magic-14,pavicic-quantum-23}. In~particular,
in~\cite{magic-14}, a~graph is obtained that has a
one-to-one correspondence with the non-KS 30-108 NBMMPH,
shown in \mbox{Figure~\ref{fig:stabilizers}(a)} and obtained
in ref. (\cite{pavicic-quantum-23}, Section~5.4). Its filled
MMPH 232-108, shown in Figure~\ref{fig:stabilizers}(b), is a KS NBMMPH,
but the latter is not critical. It contains just one critical KS
MMPH: the 152-71 one. Notice that it too violates the
$\alpha$-inequality $\alpha=64>\alpha^*=38$
(\cite{pavicic-quantum-23}, Table~5) but~is nevertheless contextual.
When its $m=1$ vertices are dropped, it becomes the 24-71 non-KS
NBMMPH shown in Figure~\ref{fig:stabilizers}(c), which satisfies the
$\alpha$-inequality: $\alpha=5<\alpha^*=6$. Thus,
notwithstanding the $\alpha$-inequality, all these MMPHs are
NBMMPHs, i.e.,~contextual, and~therefore might provide a
path toward finding simpler stabilizer operations by reducing
the original noncritical stabilizer NBMMPHs to smaller,
critical ones.

\begin{figure}[H]
\center
  \includegraphics[width=0.95\textwidth]{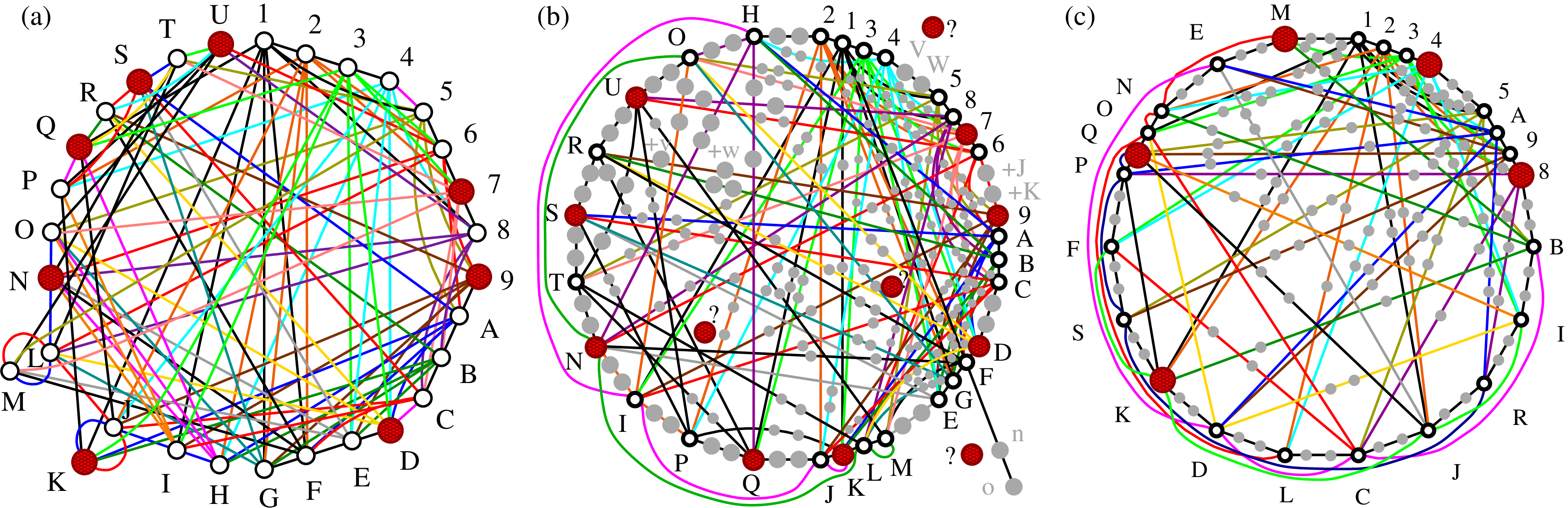}
  \caption{\baselineskip=10.5pt (\textbf{a}) Stabilizer 
 non-KS 30-108 NBMMPH~\cite{magic-14}:
  $\alpha=8<\alpha^*=9$; (\cite{pavicic-quantum-23}, Section~5.4);
  (\textbf{b}) filled (\textbf{a}): noncritical 232-108 KS MMPH
  violates the $\alpha$-inequality: $\alpha=101>\alpha^*=58$ (some
  red vertices that might contribute to $\alpha$ are indicated by
  ``?''; so, e.g., the red dot near vertex {\tt n} at the right
  bottom of the figure indicate that {\tt n} might contribute to
  $\alpha$); (\textbf{c}) critical 152-71 that violates the
  $\alpha$-inequality: $\alpha=64>\alpha^*=38$; with $m=1$ (grey)
  vertices dropped, we have $\overline{\rm subhypergraph}$ 24-71,
  which satisfies the $\alpha$-inequality: $\alpha=5<\alpha^*=6$.}
\label{fig:stabilizers}
\end{figure}

\section{\label{sec:disc}Discussion}

We present a survey and further development,
as well as promising applications, of quantum contextual
sets represented by a type of hypergraph called
McKay--Megill--Pavi\v ci\'c hypergraphs (MMPHs; Definition
\ref{def:MMP-string}). Such a representation is universal as
it can be demonstrated that most presentations of contextual
sets---including polytopes, Lie groups, operators, projectors,
vectors, and~states---have a one-to-one correspondence with
MMPHs~\cite{pavicic-quantum-23}. A~good outline of such
alternative methods is given in ref.~\cite{pwma-19}.

In this paper, we focus on two methods of obtaining contextual
MMPHs: the vector generation of MMPHs from basic vector components
via the {\bf M1} method (see Section~\ref{subsec:nonks}) and
the dimensional upscaling of MMPHs via the {\bf M8} method.

The former method involves creating an internal list of all
possible non-zero vectors containing simple vector components
like $\{0,\pm 1\}$ or $\{0,\pm $i$\}$. From~this list, a~program
(VECFIND) finds a set of all possible mutually orthogonal vector
$n$-tuples of vectors and, from them, it generates an MMPH with
hyperedges corresponding to mutually orthogonal $n$-tuples, which
is called a master MMPH. Contextual MMPHs, i.e.,~NBMMPHs
(\mbox{Definition~\ref{def:n-b}}), are then filtered out by the
program STATES01. These NBMMPHs allow us to obtain a class
(Definition~\ref{df:multi}) of critical
(Definition~\ref{def:critical}) KS MMPHs
(Th.~\ref{def:ks-theorem}) as well as the~non-KS MMPHs
(Definition~\ref{def:nonKSdef}) via {\bf M3} and {\bf M5}.

The latter method consists of using comparatively small KS MMPHs
obtained with the help of the former method to build KS MMPHs
in higher dimensions, which then allow us to recursively obtain
KS MMPHs in increasingly higher dimensions, since the complexity of
the method does not scale with dimension. This~method works by
combining two KS MMPHs from lower dimensions so as to allow
the number of unique vertices in the new combined MMPH to
be minimized and filter out KS MMPHs. We start with $k_1$ vectors
in the $n_1$-dim space and $k_2$ vectors in the $n_2$-dim space and
construct an MMPH with $n$-dim vectors, where $n<n_1+n_2$---vectors
of parent MMPHs are extended with enough vector components
0 to reach dimension $n$. In~the first MMPH, 0s are appended at
the ends of the existing vectors. In~the second MMPH,
the~$n$-$n_2$-dim 0s are prefixed at the start of the vectors, which
ensures that the new MMPH has non-zero vector components in all $n$
cardinal directions in the space. Since the vector components
$\{0,\pm 1\}$ are used, the~number of vectors $k$ of the $n$-dim
space is significantly smaller than $k_1+k_2$ for most parent KS
MMPHs, and almost all of them are KS~MMPHs.

In Sections~\ref{subsubsec:3D}, \ref{subsubsec:4D}, and~\ref{subsubsec:5-8D}, we explore generations of MMPHs in dimensions
three to eight from simple vector components, such as $\{0,\pm1\}$,
and discuss specific new instances that have not previously
been considered. We distinguish two main groups of contextual MMPHs:
Kochen--Specker (KS) MMPHs (Definition~\ref{def:KSdef}) and non-KS MMPHs
(Definition~\ref{def:nonKSdef}). The~KS MMPHs are generated from large
master MMPHs and yield non-KS MMPHs. However, the~generation of
these masters and subsequent smaller MMPHs from them is a task
of exponential complexity, which limits the methods of obtaining them
(Methods {\bf M1,2}, Section~\ref{subsec:nonks}) to dimension eight.
To overcome this limitation, we have developed methods {\bf M7} and
{\bf M8}, whose complexity does not scale with dimension,
and we elaborate on them in dimensions 9 to 32 in
Section~\ref{subsubsec:9-32D} while presenting a new graphical
representation of MMPHs that offers a more comprehensive insight
into their structure. The~minimal number of hyperedges in
non-KS MMPHs repeatedly varies between 8 (odd dimensions)
and 9 (even dimensions). For~KS NBMMPHs, this number
fluctuates between 9 and 16. We provide lists of the NBMMPHs
across all dimensions in Tables~\ref{T:3D}--\ref{T:small2}.

The extensive range of MMPHs has allowed us to identify several
features that they exhibit, as~well as the relations and quantifications
of contextuality that they support. In~Section~\ref{sec:results}, we present
several lemmas and theorems related to these features, along with
two types of statistics (Definitions~\ref{def:raw} and \ref{def:post}).
We also prove that the $\alpha$-inequality between the independence
number ${\alpha}$ (Definition~\ref{def:alpha}) and the
fractional independence (packing) number $\alpha^*$
(Definition~\ref{def:alpha-star-Wolfram}), expressed as
${\alpha}\le\alpha^*$ (Equation~(\ref{eq:alpha-cabelllo})),
does not always hold for quantum
measurements of MMPH states. The~reason is that, in quantum
measurements, all vertices of a hyperedge must be assigned an equal
and constant probability when measured, which leads us to the
Quantum Indeterminacy Postulate \ref{postulate} and two types of
statistics, {\em raw} MMPH {\em data statistics}
(Definition~\ref{def:raw}) and {\em postprocessed} MMPH {\em data
statistics} (Definition~\ref{def:post}), as~well as to two noncontextual
inequalities, the {\em hyperedge inequality} (Equation~(\ref{eq:e-ineq}))
and the {\em postprocessed quantum fractional independence number
inequality} (Equation~(\ref{eq:alpha-alpha-b})), which serve as reliable
discriminators of contextual MMPHs, in contrast to the
$\alpha$-inequality, which does not. In~Figure~\ref{fig:alpha-star},
we show several KS and non-KS MMPHs that violate the
$\alpha$-inequality. Hence, the~$\alpha$-inequality in, e.g.,
\cite{magic-14,cabello-severini-winter-14} does not
{\em reveal} contextuality but~is fortuitously satisfied for
the MMPHs that might already be known as contextual, i.e.,~as
NBMMPHs, what can be reliably verified with the help of our
programs STATES01 and
ONE---cf.~Figure~\ref{fig:stabilizers}(a,b)---by
processing measurement~data.

In Section~\ref{subsec:app}, we propose four possible applications
of contextual~MMPHs.

In Section~\ref{appA}, we consider larger alphabet quantum key
distribution (QKD) protocols in higher dimensions. We note
that conventional quantum contextual protocols, as~proposed in
sources like~\cite{cabello-dambrosio-11}, do not offer
a quantum advantage over the corresponding classical protocols,
due to the impossibility of the existence of such competitive
classical protocols by virtue of the Kochen--Specker theorem.
Therefore, we design a protocol that does provide a quantum
advantage over possible classical competitors such as BB84,
even when expanded to larger alphabets. The~protocol relies
on the Quantum Indeterminacy \mbox{Postulate \ref{postulate}} and
genuine quantum~randomness. 

In Section~\ref{appB}, we extend the previous protocol to a
quantum oblivious communication protocol in higher~dimensions. 

In Section~\ref{appC}, we examine the $S$ class of the Hadamard
transform, which is a special case of the Fourier transform that
underlies most quantum computation algorithms. We establish a
one-to-one correspondence between the elements of the $S$ class
and star-like KS MMPHs, thereby allowing us to derive elements
of an S matrix from a corresponding star-like KS~MMPH.

In Section~\ref{appD}, we explore the application of NBMMPHs in
simplifying stabilizer operations. Specifically, we consider a
recent derivation of a graph from a stabilizer operation and
demonstrate that translating this graph into a non-KS MMPH allows
for the addition of $m=1$ vertices to form a KS MMPH. The~resulting
critical KS MMPH and~its underlying non-KS MMPH (with $m=1$
vertices removed) are significantly smaller than the original
versions. This approach could lead to a reduced set of stabilizer
operations and facilitate the discovery of simpler error correction
codes.

\section{\label{sec:met}Methods}

The methods that we employ to manage quantum contextual sets are
based on algorithms and programs developed within the MMP language, 
including VECFIND, STATES01, MMPSTRIP, MMPSHUFFLE, SUBGRAPH, LOOP, 
SHORTD, and~ONE, as referenced in
~\cite{pavicic-quantum-23}.
These resources are freely accessible
at \url{http://puh.srce.hr/s/Qegixzz2BdjYwFL} (accessed on 6 January 2025). 

MMPHs can be visualized as hypergraph figures, consisting of dots
and lines, and~can also be represented as strings of ASCII
characters. This representation enables simultaneous
processing of billions of MMPHs using supercomputers and clusters.
To facilitate the processing, we have developed additional
dynamical programs to manage and parallelize tasks involving any
number of MMPH vertices and~edges. 

\section{\label{sec:concl}Conclusions}

In summary, based on elaborations of KS and non-KS contextual sets
presented in the literature and further developed in this paper, we
have developed methods for the generation of contextual sets,
revealed their properties, and~designed their applications across
any dimension. We provide examples in dimensions up to 32 and
give reliable discriminators of their contextuality. A~more
detailed summary of the achieved results is given in
Section~\ref{sec:disc}.

\bigskip\bigskip
{\em Acknowledgements}---Supported by the Ministry of Science
and Education of Croatia through the Center of Excellence for
Advanced Materials and Sensing Devices (CEMS) funding and~by
MSE grants No. KK.01.1.1.01.0001 and 533-19-15-0022.
Also supported by the Humboldt Foundation, Germany. 
Computational support was provided by the Zagreb University
Computing Centre.The technical support of Emir
Imamagi\'c from the University of Zagreb Computing Centre 
is gratefully acknowledged.

\bigskip\bigskip
{\em Abbreviations}---The following abbreviations are used
  in this manuscript:

\noindent 
\begin{tabular}{@{}ll}
  MMPH & McKay--Megill--Pavi\v ci\'c hypergraph
         (Definition \ref{def:MMP-string})\\
  NBMMPH & Non-binary McKay--Megill--Pavi\v ci\'c hypergraph
         (Definition \ref{def:n-b})\\
  BMMPH & Binary McKay--Megill--Pavi\v ci\'c hypergraph
         (Definition \ref{def:bin})\\
  KS & Kochen--Specker (Definition \ref{def:KSdef})\\
  non-KS & Non-Kochen--Specker (Definition \ref{def:nonKSdef})\\
  {\bf M1, M2, \dots\ , M8} & Methods 1, 2, \dots\ , 8 (Section
         \ref{subsec:nonks})\\
\end{tabular}

\bigskip
\appendix
\section[\appendixname~\thesection]{\label{app:0} ASCII
  Strings of Non-KS MMPH Classes and
  Their Masters and Supermasters}

   Below, we give the strings and coordinatizations of all MMPHs referred
   to in the main body of the paper. The~first hyperedges in a line
   of a critical NBMMPH often correspond to the largest loops in
   the~figures.

\subsection{\label{app:1} 3-dim~MMPHs}

\parindent=0pt

{\bf 8-7} (KS ``bug'') 123,34,45,567,78,81,26. 

\smallskip

{\bf 13-7} (filled$\,\,${\bf 8-7}) {\tt 123,394,4A5,567,7B8,8C1,2D6.} {\tt 1} = (0,0,1), {\tt 2} = (0,1,0), {\tt 3} = (1,0,0), {\tt 4} = (0,1,1),\linebreak  {\tt 5} = (1,1,$-$1), {\tt 6} = (1,0,1), {\tt 7} = ($-$1,2,1), {\tt 8} = (2,1,0), {\tt 9} = (0,1,$-$1), {\tt A} = (2,$-$1,1), {\tt B} = (1,$-$2,5),\linebreak  {\tt C} = (1,$-$2,0), {\tt D} = (1,0,$-$1).

\smallskip

{{\bf 25-16} (BMMPH) {\tt 123,345,567,789,9AB,BCD,DEF,FGH,HI1,1JK,KLA,JM7,JPD,HN9,3OB,FL5.}\linebreak  {\tt 1} = (0,0,1), {\tt 2} = ($\sqrt{2}$,$-$1,0), {\tt 3} = (1,$\sqrt{2}$,0), {\tt 4} = ($\sqrt{2}$,$-$1,$-$3), {\tt 5} = ($\sqrt{2}$,$-$1,1), {\tt 6} = ($\sqrt{2}$,3,1),\linebreak  {\tt 7} = ($-$1,0,$\sqrt{2}$), {\tt 8} = ($\sqrt{2}$,$-$3,1), {\tt 9} = ($\sqrt{2}$,1,1), {\tt A} = (0,1,$-$1), {\tt B} = ($\sqrt{2}$,$-$1,$-$1), {\tt C} = ($\sqrt{2}$,3,$-$1),\linebreak  {\tt D} = (1,0,$\sqrt{2}$), {\tt E} = ($\sqrt{2}$,$-$3,$-$1), {\tt F} = ($\sqrt{2}$,1,$-$1), {\tt G} = ($\sqrt{2}$,1,3), {\tt H} = ($-$1,$\sqrt{2}$,0), {\tt I} = ($\sqrt{2}$,1,0),\linebreak  {\tt J} = (0,1,0), {\tt K} = (1,0,0), {\tt L} = (0,1,1), {\tt M} = ($\sqrt{2}$,0,1), {\tt N} = ($\sqrt{2}$,1,$-$3), {\tt O} = ($\sqrt{2}$,$-$1,3), {\tt P} = ($\sqrt{2}$,0,$-$1).}

\smallskip
{\bf 95-66} {\tt 123,145,167,189,1AB,1CD,1EF,2GH,2IJ,2KL,2MN,2OP,2QR,ST3,SUV,WX3,WYZ,Wab,\linebreak Xcd,Xef,KgU,Khb,MiY,Mjk,lQa,mEe,nEo,pAc,qAr,stu,svL,sw5,xQy,z!L,z"7,t\#N,t\$7,8\%\&,\linebreak 8'f,()*, (-9,(/H,:;9,:\textless J,=)B,\textgreater ?B,)@J,?/\textless ,[$\backslash$N,]C\^{},\_Cd,`OV,\{OZ,$\backslash$w",|"R, \}\textless F,$\sim$*F,\linebreak +1/D,+2@D,+3Ru, +4wP,+5\$P,kaV,yYU,re\^{},co\&.} {\tt 1} = (0,0,1), {\tt 2} = (0,1,0), {\tt S} = (0,1,$-2^{\frac{1}{2}}\phi^{\frac{3}{2}}$),\linebreak  {\tt W} = (0,1,$2^{\frac{1}{2}}\phi^{-\frac{1}{2}}$), {\tt T} = (0,$2^{\frac{1}{2}}\phi^{\frac{3}{2}}$,1), {\tt X} = (0,$-2^{\frac{1}{2}}\phi^{-\frac{1}{2}}$,1),  {\tt 3} = (1,0,0),  {\tt G} = (1,0,$2^{\frac{1}{2}}\phi^{\frac{3}{2}}$),\hfil  {\tt I} = (1,0,$-2^{\frac{1}{2}}\phi^{-\frac{1}{2}}$),\linebreak   {\tt K}=(1,0,$5^{\frac{1}{4}}\phi^{\frac{3}{2}}$), {\tt M}=(1,0,$5^{\frac{1}{4}}\phi^{-\frac{3}{2}}$), {\tt i}=($2^{-1}5^{\frac{1}{4}}\phi^{-\frac{5}{2}}$,$-2^{\frac{1}{2}}\phi^{-\frac{5}{2}}$,$-2^{-1}\phi^{-1}$), {\tt l}=($2^{-1}5^{\frac{1}{4}}\phi^{-\frac{5}{2}}$,$2^{\frac{1}{2}}\phi^{-\frac{5}{2}}$,$2^{-1}\phi^{-1}$),\linebreak {\tt m} = ($2^{-1}5^{\frac{1}{4}}\phi^{-\frac{5}{2}}$,$2^{-1}\phi^{-1}$,$-2^{\frac{1}{2}}\phi^{-\frac{5}{2}}$), {\tt n} = ($2^{-1}5^{\frac{1}{4}}\phi^{-\frac{5}{2}}$,$2^{-1}\phi^{-1}$,$2^{\frac{1}{2}}\phi^{-\frac{1}{2}}$), {\tt 4} = (1,$2^{\frac{1}{2}}\phi^{\frac{3}{2}}$,0),\linebreak {\tt p} = ($2^{-1}5^{\frac{1}{4}}\phi^{-\frac{5}{2}}$,$-2^{-1}\phi^{-1}$,$2^{\frac{1}{2}}\phi^{-\frac{5}{2}}$), {\tt q} = ($2^{-1}5^{\frac{1}{4}}\phi^{-\frac{5}{2}}$,$-2^{-1}\phi^{-1}$,$-2^{\frac{1}{2}}\phi^{-\frac{1}{2}}$), {\tt s} = (1,$2^{\frac{1}{2}}\phi^{\frac{3}{2}}$,$5^{\frac{1}{4}}\phi^{\frac{3}{2}}$),\linebreak  {\tt j} = ($2^{-1}5^{\frac{1}{4}}\phi^{-\frac{5}{2}}$,$2^{\frac{1}{2}}\phi^{-\frac{1}{2}}$,$-2^{-1}\phi^{-1}$), {\tt x} = ($2^{-1}5^{\frac{1}{4}}\phi^{-\frac{5}{2}}$,$-2^{\frac{1}{2}}\phi^{-\frac{1}{2}}$,$2^{-1}\phi^{-1}$), {\tt 6} = (1,$-2^{\frac{1}{2}}\phi^{-\frac{1}{2}}$,0),\linebreak {\tt z} = (1,$-2^{\frac{1}{2}}\phi^{-\frac{1}{2}}$,$5^{\frac{1}{4}}\phi^{\frac{3}{2}}$), {\tt t} = (1,$-2^{\frac{1}{2}}\phi^{-\frac{1}{2}}$,$5^{\frac{1}{4}}\phi^{-\frac{3}{2}}$), {\tt 8} = (1,$5^{\frac{1}{4}}\phi^{\frac{3}{2}}$,0), {\tt (} = (1,$5^{\frac{1}{4}}\phi^{\frac{3}{2}}$,$2^{\frac{1}{2}}\phi^{\frac{3}{2}}$),\linebreak {\tt :} = (1,$5^{\frac{1}{4}}\phi^{\frac{3}{2}}$,$-2^{\frac{1}{2}}\phi^{-\frac{1}{2}}$), \ {\tt  =} = ($2^{-1}\phi^{-1}$,$2^{-1}5^{\frac{1}{4}}\phi^{-\frac{5}{2}}$,$2^{\frac{1}{2}}\phi^{-\frac{5}{2}}$), {\tt \textgreater } = ($2^{-1}\phi^{-1}$,$2^{-1}5^{\frac{1}{4}}\phi^{-\frac{5}{2}}$,$-2^{\frac{1}{2}}\phi^{-\frac{1}{2}}$),\linebreak  {\tt A} = (1,$5^{\frac{1}{4}}\phi^{-\frac{3}{2}}$,0), {\tt \#} = ($2^{-1}\phi^{-1}$,$2^{\frac{1}{2}}\phi^{-\frac{5}{2}}$,$2^{-1}5^{\frac{1}{4}}\phi^{-\frac{5}{2}}$), {\tt )} = (1,$5^{\frac{1}{4}}\phi^{-\frac{3}{2}}$,$-2^{\frac{1}{2}}\phi^{-\frac{1}{2}}$),\linebreak  {\tt -} = ($2^{-1}\phi^{-1}$,$2^{-1}5^{\frac{1}{4}}\phi^{\frac{1}{2}}$,$-2^{\frac{1}{2}}\phi^{-\frac{1}{2}}$), {\tt ;} = ($2^{-1}\phi^{-1}$,$2^{-1}5^{\frac{1}{4}}\phi^{\frac{1}{2}}$,$2^{\frac{1}{2}}\phi^{\frac{3}{2}}$), {\tt ?} = (1,$5^{\frac{1}{4}}\phi^{-\frac{3}{2}}$,$2^{\frac{1}{2}}\phi^{-\frac{5}{2}}$),\linebreak  {\tt [} = ($2^{-1}\phi^{-1}$,$-2^{\frac{1}{2}}\phi^{-\frac{1}{2}}$,$2^{-1}5^{\frac{1}{4}}\phi^{-\frac{5}{2}}$),\hfil  {\tt v} = ($2^{-1}\phi^{-1}$,$-2^{\frac{1}{2}}\phi^{-\frac{1}{2}}$,$2^{-1}5^{\frac{1}{4}}\phi^{\frac{1}{2}}$), {\tt !} = ($2^{-1}\phi^{-1}$,$2^{\frac{1}{2}}\phi^{\frac{3}{2}}$,$2^{-1}5^{\frac{1}{4}}\phi^{\frac{1}{2}}$),\linebreak  {\tt ]} = ($2^{-1}5^{\frac{1}{4}}\phi^{\frac{1}{2}}$,$2^{-1}\phi^{-1}$,$2^{\frac{1}{2}}\phi^{-\frac{1}{2}}$), {\tt \_} = ($2^{-1}5^{\frac{1}{4}}\phi^{\frac{1}{2}}$,$2^{-1}\phi^{-1}$,$-2^{\frac{1}{2}}\phi^{\frac{3}{2}}$), {\tt \%} = ($2^{-1}5^{\frac{1}{4}}\phi^{\frac{1}{2}}$,$-2^{-1}\phi^{-1}$,$-2^{\frac{1}{2}}\phi^{-\frac{1}{2}}$),\linebreak   {\tt '} = ($2^{-1}5^{\frac{1}{4}}\phi^{\frac{1}{2}}$,$-2^{-1}\phi^{-1}$,$2^{\frac{1}{2}}\phi^{\frac{3}{2}}$), {\tt g} = ($2^{-1}5^{\frac{1}{4}}\phi^{\frac{1}{2}}$,$2^{\frac{1}{2}}\phi^{-\frac{1}{2}}$,$-2^{-1}\phi^{-1}$), {\tt h} = ($2^{-1}5^{\frac{1}{4}}\phi^{\frac{1}{2}}$,$-2^{\frac{1}{2}}\phi^{\frac{3}{2}}$,$-2^{-1}\phi^{-1}$),\linebreak {\tt `} = ($2^{-1}5^{\frac{1}{4}}\phi^{\frac{1}{2}}$,$-2^{\frac{1}{2}}\phi^{-\frac{1}{2}}$,$2^{-1}\phi^{-1}$), {\tt \{} = ($2^{-1}5^{\frac{1}{4}}\phi^{\frac{1}{2}}$,$2^{\frac{1}{2}}\phi^{\frac{3}{2}}$,$2^{-1}\phi^{-1}$), {\tt $\backslash$} = (1,$2^{\frac{1}{2}}\phi^{-\frac{5}{2}}$,$5^{\frac{1}{4}}\phi^{-\frac{3}{2}}$),\linebreak {\tt |} = ($-2^{-1}\phi^{-1}$,$-2^{\frac{1}{2}}\phi^{-\frac{5}{2}}$,$2^{-1}5^{\frac{1}{4}}\phi^{-\frac{5}{2}}$),\hfil {\tt \}} = ($-2^{-1}\phi^{-1}$,$2^{-1}5^{\frac{1}{4}}\phi^{-\frac{5}{2}}$,$-2^{\frac{1}{2}}\phi^{-\frac{5}{2}}$),\linebreak  {\tt $\sim$} = ($-2^{-1}\phi^{-1}$,$2^{-1}5^{\frac{1}{4}}\phi^{-\frac{5}{2}}$,$2^{\frac{1}{2}}\phi^{-\frac{1}{2}}$),\hfil  {\tt +1} = ($-2^{-1}\phi^{-1}$,$2^{-1}5^{\frac{1}{4}}\phi^{\frac{1}{2}}$,$2^{\frac{1}{2}}\phi^{-\frac{1}{2}}$),\linebreak  {\tt +2} = ($-2^{-1}\phi^{-1}$,$2^{-1}5^{\frac{1}{4}}\phi^{\frac{1}{2}}$,$-2^{\frac{1}{2}}\phi^{\frac{3}{2}}$),\hfil {\tt +3} = ($-2^{-1}\phi^{-1}$,$2^{\frac{1}{2}}\phi^{-\frac{1}{2}}$,$2^{-1}5^{\frac{1}{4}}\phi^{-\frac{5}{2}}$),\linebreak  {\tt +4} = ($-2^{-1}\phi^{-1}$,$2^{\frac{1}{2}}\phi^{-\frac{1}{2}}$,$2^{-1}5^{\frac{1}{4}}\phi^{\frac{1}{2}}$),\hfil {\tt +5} = ($-2^{-1}\phi^{-1}$,$-2^{\frac{1}{2}}\phi^{\frac{3}{2}}$,$2^{-1}5^{\frac{1}{4}}\phi^{\frac{1}{2}}$),\hfil {\tt O} = ($-$1,0,$5^{\frac{1}{4}}\phi^{\frac{3}{2}}$),\linebreak {\tt Q} = ($-$1,0,$5^{\frac{1}{4}}\phi^{-\frac{3}{2}}$),\hfil  {\tt C} = ($-$1,$5^{\frac{1}{4}}\phi^{\frac{3}{2}}$,0),\hfil  {\tt /} = ($-$1,$5^{\frac{1}{4}}\phi^{\frac{3}{2}}$,$-2^{\frac{1}{2}}\phi^{\frac{3}{2}}$),\hfil  {\tt @} = ($-$1,$5^{\frac{1}{4}}\phi^{\frac{3}{2}}$,$2^{\frac{1}{2}}\phi^{-\frac{1}{2}}$),\linebreak {\tt k} = ($5^{\frac{1}{4}}\phi^{-\frac{3}{2}}$,$-2^{\frac{1}{2}}\phi^{-\frac{5}{2}}$,$-$1),\hfil  {\tt w} = ($-$1,$-2^{\frac{1}{2}}\phi^{\frac{3}{2}}$,$5^{\frac{1}{4}}\phi^{\frac{3}{2}}$),\hfil  {\tt \$} = ($-$1,$2^{\frac{1}{2}}\phi^{-\frac{1}{2}}$,$5^{\frac{1}{4}}\phi^{\frac{3}{2}}$),\linebreak {\tt "} = ($-$1,$2^{\frac{1}{2}}\phi^{-\frac{1}{2}}$,$5^{\frac{1}{4}}\phi^{-\frac{3}{2}}$), {\tt y} = ($5^{\frac{1}{4}}\phi^{-\frac{3}{2}}$,$2^{\frac{1}{2}}\phi^{-\frac{5}{2}}$,1), {\tt E} = ($-$1,$5^{\frac{1}{4}}\phi^{-\frac{3}{2}}$,0), {\tt \textless } = ($-$1,$5^{\frac{1}{4}}\phi^{-\frac{3}{2}}$,$2^{\frac{1}{2}}\phi^{-\frac{1}{2}}$),\linebreak {\tt Y} = ($5^{\frac{1}{4}}\phi^{-\frac{3}{2}}$,$2^{\frac{1}{2}}\phi^{-\frac{1}{2}}$,$-$1),\hfil  {\tt *} = ($-$1,$5^{\frac{1}{4}}\phi^{-\frac{3}{2}}$,$-2^{\frac{1}{2}}\phi^{-\frac{5}{2}}$),\hfil  {\tt a} = ($5^{\frac{1}{4}}\phi^{-\frac{3}{2}}$,$-2^{\frac{1}{2}}\phi^{-\frac{1}{2}}$,1),\linebreak {\tt r} = ($5^{\frac{1}{4}}\phi^{-\frac{3}{2}}$,$-$1,$2^{\frac{1}{2}}\phi^{-\frac{5}{2}}$), {\tt c} = ($5^{\frac{1}{4}}\phi^{-\frac{3}{2}}$,$-$1,$-2^{\frac{1}{2}}\phi^{-\frac{1}{2}}$), {\tt B} = ($5^{\frac{1}{4}}\phi^{-\frac{3}{2}}$,$-$1,0), {\tt o} = ($5^{\frac{1}{4}}\phi^{-\frac{3}{2}}$,1,$-2^{\frac{1}{2}}\phi^{-\frac{5}{2}}$), {\tt e} = ($5^{\frac{1}{4}}\phi^{-\frac{3}{2}}$,1,$2^{\frac{1}{2}}\phi^{-\frac{1}{2}}$),\hfil  {\tt F} = ($5^{\frac{1}{4}}\phi^{-\frac{3}{2}}$,1,0),\hfil  {\tt N} = ($5^{\frac{1}{4}}\phi^{-\frac{3}{2}}$,0,$-$1),\hfil  {\tt R} = ($5^{\frac{1}{4}}\phi^{-\frac{3}{2}}$,0,1),\linebreak {\tt u} = ($-$1,$-2^{\frac{1}{2}}\phi^{-\frac{5}{2}}$,$5^{\frac{1}{4}}\phi^{-\frac{3}{2}}$), {\tt H} = ($-2^{\frac{1}{2}}\phi^{\frac{3}{2}}$,0,1), {\tt 5} = ($-2^{\frac{1}{2}}\phi^{\frac{3}{2}}$,1,0), {\tt b} = ($5^{\frac{1}{4}}\phi^{\frac{3}{2}}$,$2^{\frac{1}{2}}\phi^{-\frac{1}{2}}$,$-$1),\linebreak {\tt U} = ($5^{\frac{1}{4}}\phi^{\frac{3}{2}}$,$-2^{\frac{1}{2}}\phi^{\frac{3}{2}}$,$-$1), {\tt Z} = ($5^{\frac{1}{4}}\phi^{\frac{3}{2}}$,$-2^{\frac{1}{2}}\phi^{-\frac{1}{2}}$,1), {\tt V} = ($5^{\frac{1}{4}}\phi^{\frac{3}{2}}$,$2^{\frac{1}{2}}\phi^{\frac{3}{2}}$,1), {\tt f} = ($5^{\frac{1}{4}}\phi^{\frac{3}{2}}$,$-$1,$-2^{\frac{1}{2}}\phi^{-\frac{1}{2}}$),\linebreak  {\tt \&} = ($5^{\frac{1}{4}}\phi^{\frac{3}{2}}$,$-$1,$2^{\frac{1}{2}}\phi^{\frac{3}{2}}$), {\tt 9} = ($5^{\frac{1}{4}}\phi^{\frac{3}{2}}$,$-$1,0), {\tt d} = ($5^{\frac{1}{4}}\phi^{\frac{3}{2}}$,1,$2^{\frac{1}{2}}\phi^{-\frac{1}{2}}$), {\tt \^{}} = ($5^{\frac{1}{4}}\phi^{\frac{3}{2}}$,1,$-2^{\frac{1}{2}}\phi^{\frac{3}{2}}$),\linebreak  {\tt D} = ($5^{\frac{1}{4}}\phi^{\frac{3}{2}}$,1,0),\hfill  {\tt L} = ($5^{\frac{1}{4}}\phi^{\frac{3}{2}}$,0,$-$1),\hfill  {\tt P} = ($5^{\frac{1}{4}}\phi^{\frac{3}{2}}$,0,1),\hfill  {\tt 7} = ($2^{\frac{1}{2}}\phi^{-\frac{1}{2}}$,1,0),\hfill  {\tt J} = ($2^{\frac{1}{2}}\phi^{-\frac{1}{2}}$,0,1).

\subsection{\label{app:2} 4-dim~MMPHs}

\medskip
{\bf 10-7} \qquad {\tt 1234,56,764,89A7,853,91,A2.}

\medskip
{\bf 24-13} {\tt LMNO,HIJK,EFGK,BCDG,9ADJ,78NO,5678,56AC,34BF,249I,1234,6EHM,146L}.\linebreak  {\tt 1} = (0,1,$-$1,0), {\tt 2} = (1,0,0,1), {\tt 3} = (1,0,0,$-$1), {\tt 4} = (0,1,1,0), {\tt 5} = (0,1,0,0), {\tt 6} = (0,0,0,1), {\tt 7} = (1,0,$i$,0), {\tt 8} = ($i$,0,1,0), {\tt 9} = (1,$-$1,1,$-$1), {\tt A} = (1,0,$-$1,0), {\tt B} = (1,1,$-$1,1), {\tt C} = (1,0,1,0), {\tt D} = (0,1,0,$-$1),\linebreak  {\tt E} = (1,1,0,0), {\tt F} = (1,$-$1,1,1), {\tt G} = ($-$1,1,1,1), {\tt H} = (1,$-$1,0,0), {\tt I} = (1,1,$-$1,$-$1), {\tt J} = (1,1,1,1),\linebreak  {\tt K} = (0,0,1,$-$1), {\tt L} = (1,0,0,0), {\tt M} = (0,0,1,0), {\tt N} = (0,1,0,$i$), {\tt O} = (0,$i$,0,1).

\bigskip
{\bf 25-15} {\tt EC1N,NOPQ,PQLM,MG5A,A98K,KJ67,6734,4HBE,ECDB,ED59,12KO,LMGF,MF4,\linebreak 8234,KJH.} {\tt 1} = (1,$-$1,1,1), {\tt 2} = (1,0,$-$1,0), {\tt 3} = (0,1,0,0), {\tt 4} = (0,0,0,1), {\tt 5} = (1,$-$1,1,$-$1),\linebreak  {\tt 6} = ($i$,0,1,0), {\tt 7} = (1,0,$i$,0), {\tt 8} = (1,0,1,0), {\tt 9} = (1,1,$-$1,$-$1), {\tt A} = (1,$-$1,$-$1,1), {\tt B} = (0,1,$-$1,0),\linebreak  {\tt C} = (1,0,0,$-$1), {\tt D} = (1,0,0,1), {\tt E} = (0,1,1,0), {\tt F} = (1,$-$1,0,0), {\tt G} = (0,0,1,1), {\tt H} = (1,0,0,0),\linebreak  {\tt J} = (0,1,0,$-$1), {\tt K} = (0,1,0,1), {\tt L} = (0,0,1,$-$1), {\tt M} = (1,1,0,0), {\tt N} = (1,1,$-$1,1), {\tt O} = (1,1,1,$-$1),\linebreak  {\tt P} = (1,$-$1,$i$,$i$), {\tt Q} = ($-$1,1,$i$,$i$).

\bigskip
{\bf 29-16} {\tt QRST,MNOP,IJKL,EFGH,BCDL,9ADP,8ACT,7BGH,69FH,58EH,347L,246P,145T,\linebreak 12BK,139O,238S.} {\tt 1} = (1,1,1,1), {\tt 2} = (1,1,$-$1,$-$1), {\tt 3} = (1,$-$1,1,$-$1), {\tt 4} = (1,$-$1,$-$1,1),\linebreak  {\tt 5} = (1,0,0,$-$1), {\tt 6} = (1,0,1,0), {\tt 7} = (0,0,1,1), {\tt 8} = (1,0,0,1), {\tt 9} = (1,0,$-$1,0), {\tt A} = (1,1,1,$-$1),\linebreak  {\tt B} = (0,0,1,$-$1), {\tt C} = ($-$1,1,1,1), {\tt D} = (1,$-$1,1,1), {\tt E} = (0,0,1,0), {\tt F} = (0,0,0,1), {\tt G} = (1,0,0,0),\linebreak  {\tt H} = (0,1,0,0), {\tt I} = (0,0,2,$-$1), {\tt J} = (0,0,1,2), {\tt K} = (1,$-$1,0,0), {\tt L} = (1,1,0,0), {\tt M} = (1,0,2,0),\linebreak  {\tt N} = (2,0,$-$1,0), {\tt O} = (0,1,0,$-$1), {\tt P} = (0,1,0,1), {\tt Q} = (1,0,0,2), {\tt R} = (2,0,0,$-$1), {\tt S} = (0,1,1,0),\linebreak  {\tt T} = (0,1,$-$1,0).

\bigskip
{\bf 30-16} {\tt 1234,4567,789A,ABCD,DEFG,GHI1,1JB9,GQ68,3KM5,EPNC,IJL5,COQH,4LPN,\linebreak DOKM,FTUQ,JRS2.} {\tt 1} = (0,0,0,1), {\tt 2} = (1,0,0,0), {\tt 3} = (0,1,$-$1,0), {\tt 4} = (0,1,1,0), {\tt 5} = (1,0,0,1),\linebreak  {\tt 6} = (1,1,$-$1,$-$1), {\tt 7} = (1,$-$1,1,$-$1), {\tt 8} = (1,$-$1,$-$1,1), {\tt 9} = (1,1,0,0), {\tt A} = (0,0,1,1), {\tt B} = (1,$-$1,0,0), {\tt C} = (1,1,1,$-$1), {\tt D} = (1,1,$-$1,1), {\tt E} = ($-$1,1,1,1), {\tt F} = (0,1,0,$-$1), {\tt G} = (1,0,1,0), {\tt H} = (1,0,$-$1,0),\linebreak  {\tt I} = (0,1,0,0), {\tt J} = (0,0,1,0), {\tt K} = ($-$1,i,i,1), {\tt L} = (1,0,0,$-$1), {\tt M} = (1,i,i,$-$1), {\tt N} = (1,$-$i,i,1),\linebreak  {\tt O} = (1,$-$1,1,1),\hfill  {\tt P} = (1,i,$-$i,1),\hfil  {\tt Q} = (0,1,0,1),\hfill  {\tt R} = (0,1,0,i),\hfill  {\tt S} = (0,i,0,1),\hfill  {\tt T} = (1,0,i,0),\hfill  {\tt U} = (i,0,1,0).

\subsubsection{\label{app:2a} New 4-dim MMPH Masters,
   Their Coordinatizations, and~Their Distributions}

\medskip
{\bf 60-72} {master from $\{0,\pm 1,\phi\}$ or from
  $\{0,\pm\phi,\frac{1}{\phi}\}$}:

\medskip
{\tt 1234,1256,1278,129A,13BC,13DE,13FG,1HI4,1JK4,1LM4,23NO,23PQ,23RS,2TU4,2VW4,2XY4,\linebreak Za34,Za56,Za78,Za9A,Z5bc,Zde6,afg6,ah6i,a5jk,alm6,no34,no56,no78,no9A,pq34,pq56,\linebreak pq78,pq9A,TUBC,TUDE,TUFG,TBmc,TCdj,UrCs,UBek,UClb,UCtu,VWBC,VWDE,VWFG,HINO,HIPQ,\linebreak HIRS,HNmb,HOej,IvOw,INdk,IOlc,IOxy,JKNO,JKPQ,JKRS,fgbc,XYBC,XYDE,XYFG,LMNO,LMPQ,\linebreak LMRS,hbci,rmcs,vmbw,lmbc,dejk,mbxy,mctu.}

\medskip
Coordinatization of {\bf 60-72} obtained from $\{0,\pm 1,\phi\}$:
\medskip

{\tt 1} = (0,0,0,1), {\tt 2} = (0,0,1,0), {\tt Z} = (0,0,1,1), {\tt a} = (0,0,1,$-$1), {\tt n} = (0,0,1,$\phi$), {\tt p} = (0,0,$-$1,$\phi$),\linebreak  {\tt q} = (0,0,$\phi$,1), {\tt o} = (0,0,$\phi$,$-$1), {\tt 3} = (0,1,0,0), {\tt T} = (0,1,0,1), {\tt U} = (0,1,0,$-$1), {\tt V} = (0,1,0,$\phi$), {\tt H} = (0,1,1,0), {\tt I} = (0,1,$-$1,0), {\tt J} = (0,1,$\phi$,0), {\tt f} = ($\phi$,$\phi$,$-$1,$-$1), {\tt X} = (0,$-$1,0,$\phi$), {\tt L} = (0,$-$1,$\phi$,0), {\tt h} = ($\phi$,$\phi$,1,1),\linebreak  {\tt Y} = (0,$\phi$,0,1), {\tt W} = (0,$\phi$,0,$-$1), {\tt M} = (0,$\phi$,1,0), {\tt r} = ($\phi$,$-$1,$\phi$,$-$1), {\tt K} = (0,$\phi$,$-$1,0), {\tt v} = ($\phi$,$-$1,$-$1,$\phi$),\linebreak  {\tt 4} = (1,0,0,0), {\tt N} = (1,0,0,1), {\tt O} = (1,0,0,$-$1), {\tt P} = (1,0,0,$\phi$), {\tt B} = (1,0,1,0), {\tt C} = (1,0,$-$1,0),\linebreak  {\tt D} = (1,0,$\phi$,0), {\tt 5} = (1,1,0,0), {\tt l} = (1,1,1,1), {\tt d} = (1,1,1,$-$1), {\tt e} = (1,1,$-$1,1), {\tt m} = (1,1,$-$1,$-$1),\linebreak  {\tt g} = (1,1,$\phi$,$\phi$), {\tt 6} = (1,$-$1,0,0), {\tt j} = (1,$-$1,1,1), {\tt b} = (1,$-$1,1,$-$1), {\tt c} = (1,$-$1,$-$1,1), {\tt k} = ($-$1,1,1,1),\linebreak  {\tt 7} = (1,$\phi$,0,0), {\tt s} = (1,$\phi$,1,$\phi$), {\tt w} = (1,$\phi$,$\phi$,1), {\tt 8} = ($\phi$,$-$1,0,0), {\tt R} = ($-$1,0,0,$\phi$), {\tt F} = ($-$1,0,$\phi$,0),\linebreak  {\tt t} = ($\phi$,1,$\phi$,1),  {\tt i} = ($-$1,$-$1,$\phi$,$\phi$), {\tt 9} = ($-$1,$\phi$,0,0), {\tt u} = ($-$1,$\phi$,$-$1,$\phi$), {\tt x} = ($\phi$,1,1,$\phi$), \ {\tt y} = ($-$1,$\phi$,$\phi$,$-$1), \ {\tt S} = ($\phi$,0,0,1), \ {\tt Q} = ($\phi$,0,0,$-$1), \ {\tt G} = ($\phi$,0,1,0), \ {\tt E} = ($\phi$,0,$-$1,0), \ {\tt A} = ($\phi$,1,0,0).

\medskip
Coordinatization of {\bf 60-72} obtained from
$\{0,\pm\phi,\frac{1}{\phi}\}$:

\medskip
{\tt 1} = (0,0,0,$\phi$), {\tt 2} = (0,0,$\phi$,0), {\tt Z} = (0,0,$\phi$,$\phi$), {\tt a} = (0,0,$\phi$,$-$$\phi$), {\tt n} = (0,0,$\phi$,1/$\phi$), {\tt p} = (0,0,$-$$\phi$,1/$\phi$),\linebreak {\tt q} = (0,0,1/$\phi$,$\phi$), {\tt o} = (0,0,1/$\phi$,$-$$\phi$), {\tt 3} = (0,$\phi$,0,0), {\tt T} = (0,$\phi$,0,$\phi$), {\tt U} = (0,$\phi$,0,$-$$\phi$), {\tt V} = (0,$\phi$,0,1/$\phi$), {\tt H} = (0,$\phi$,$\phi$,0), {\tt I} = (0,$\phi$,$-$$\phi$,0), {\tt J} = (0,$\phi$,1/$\phi$,0), {\tt f} = (1/$\phi$,1/$\phi$,$-$$\phi$,$-$$\phi$), {\tt X} = (0,$-$$\phi$,0,1/$\phi$),\linebreak  {\tt L} = (0,$-$$\phi$,1/$\phi$,0), {\tt h} = (1/$\phi$,1/$\phi$,$\phi$,$\phi$), {\tt Y} = (0,1/$\phi$,0,$\phi$), {\tt W} = (0,1/$\phi$,0,$-$$\phi$), {\tt M} = (0,1/$\phi$,$\phi$,0),\linebreak  {\tt r} = (1/$\phi$,$-$$\phi$,1/$\phi$,$-$$\phi$), {\tt K} = (0,1/$\phi$,$-$$\phi$,0), {\tt v} = (1/$\phi$,$-$$\phi$,$-$$\phi$,1/$\phi$), {\tt 4} = ($\phi$,0,0,0), {\tt N} = ($\phi$,0,0,$\phi$),\linebreak  {\tt O} = ($\phi$,0,0,$-$$\phi$), {\tt P} = ($\phi$,0,0,1/$\phi$), {\tt B} = ($\phi$,0,$\phi$,0), {\tt C} = ($\phi$,0,$-$$\phi$,0), {\tt D} = ($\phi$,0,1/$\phi$,0), {\tt 5} = ($\phi$,$\phi$,0,0),\linebreak  {\tt l} = ($\phi$,$\phi$,$\phi$,$\phi$), {\tt d} = ($\phi$,$\phi$,$\phi$,$-$$\phi$), {\tt e} = ($\phi$,$\phi$
,$-$$\phi$,$\phi$), {\tt m} = ($\phi$,$\phi$,$-$$\phi$,$-$$\phi$), {\tt g} = ($\phi$,$\phi$,1/$\phi$,1/$\phi$),\linebreak  {\tt 6} = ($\phi$,$-$$\phi$,0,0), {\tt j} = ($\phi$,$-$$\phi$,$\phi$,$\phi$), {\tt b} = ($\phi$,$-$$\phi$,$\phi$,$-$$\phi$), {\tt c} = ($\phi$,$-$$\phi$,$-$$\phi$,$\phi$), {\tt k} = ($-$$\phi$,$\phi$,$\phi$,$\phi$),\linebreak  {\tt 7} = ($\phi$,1/$\phi$,0,0), {\tt s} = ($\phi$,1/$\phi$,$\phi$,1/$\phi$), {\tt w} = ($\phi$,1/$\phi$,1/$\phi$,$\phi$), {\tt 8} = (1/$\phi$,$-$$\phi$,0,0), {\tt R} = ($-$$\phi$,0,0,1/$\phi$),\linebreak {\tt F} = ($-$$\phi$,0,1/$\phi$,0), \ {\tt t} = (1/$\phi$,$\phi$,1/$\phi$,$\phi$), \ {\tt i} = ($-$$\phi$,$-$$\phi$,1/$\phi$,1/$\phi$), \ {\tt 9} = ($-$$\phi$,1/$\phi$,0,0),\linebreak  {\tt u} = ($-$$\phi$,1/$\phi$,$-$$\phi$,1/$\phi$), {\tt x} = (1/$\phi$,$\phi$,$\phi$,1/$\phi$), {\tt y} = ($-$$\phi$,1/$\phi$,1/$\phi$,$-$$\phi$), {\tt S} = (1/$\phi$,0,0,$\phi$),\linebreak  {\tt Q} = (1/$\phi$,0,0,$-$$\phi$), \ {\tt G} = (1/$\phi$,0,$\phi$,0), \ {\tt E} = (1/$\phi$,0,$-$$\phi$,0), \ {\tt A} = (1/$\phi$,$\phi$,0,0).

\medskip
{\bf 86-152} {master from $\{0,\pm 1,i\}$ or from
  $\{0,\pm i,1\}$}:

\medskip
{\tt 1234,1256,1278,139A,13BC,1DE4,1FG4,23HI,23JK,2LM4,2NO4,PQ34,PQ56,PQ78,P5RS,PTU6,\linebreak P6VW,P7XY,P7Za,Pbc8,P8de,Q5fg,Q5hi,Qjk6,Ql6m,Q7no,Q7pq,Qrs8,Qt8u,vw34,vw56,vw78,\linebreak v5xy,v5z!,v"6\#,v\$6\%,v7\&',v(8),w5*-,w5/:,w;6$<$,w=6$>$,w7?@,w[8$\backslash$,LM9A,LMBC,L9kS,LATf,\linebreak LA@),LB:\%,LBZu,L=zC,LCqe,M9Ug,M9(?,MAjR,MA]\^{},MB!$>$,MBpd,M\$/C,MtCa,NO9A,NOBC,N9x$<$,\linebreak N9cn,NA*\#,NArY,NB\&$\backslash$,NhCW,DEHI,DEJK,DHkR,DIUf,DI'$\backslash$,DJy$<$,DJau,D"*K,DKqd,EHTg,EH[\&,\linebreak EIjS,EI\_`,EJ-\#,EJpe,E;xK,EtKZ,FGHI,FGJK,FH/\%,FHbn,FIz$>$,FIrX,FJ?),FhKV,O9"-,O9sX,\linebreak OA;y,OAbo,OBiV,O[C',GH=!,GHsY,GI\$:,GIco,GJiW,G(K@,jkRS,TUfg,TUhi,Tf(?,Tg'$\backslash$,;xy$<$,\linebreak ;xau,;z!$<$,;cyn,Uf[\&,Ug@),kR\_`,kS]\^{},"*-\#,"*pe,"/:\#,"r-Y,\$*-\%,\$/:\%,\$/Zu,\$b:n,=xy$>$,\linebreak =z!$>$,=zpd,=r!X,lRSm,fgVW,*s\#X,xb$<$o,zs$>$Y,/c\%o,hiVW,rsXY,rsZa,bcno,bcpq,(?@),[\&'$\backslash$,\linebreak ty$<$Z,t:\%a,tXYu,tZau,-\#qd,!$>$qe,node,pqde.}

\medskip
Coordinatization of {\bf 86-152} obtained from $\{0,\pm 1,i\}$:

\medskip
{\tt 1} = (0,0,0,1), {\tt 2} = (0,0,1,0), {\tt P} = (0,0,1,1), {\tt Q} = (0,0,1,$-$1), {\tt v} = (0,0,1,i), {\tt w} = (0,0,i,1),\linebreak  {\tt 3} = (0,1,0,0), {\tt L} = (0,1,0,1), {\tt M} = (0,1,0,$-$1), {\tt N} = (0,1,0,i), {\tt D} = (0,1,1,0), {\tt E} = (0,1,$-$1,0),\linebreak  {\tt F} = (0,1,i,0), {\tt O} = (0,i,0,1), {\tt G} = (0,i,1,0), {\tt 4} = (1,0,0,0), {\tt H} = (1,0,0,1), {\tt I} = (1,0,0,$-$1), {\tt J} = (1,0,0,i),\linebreak  {\tt 9} = (1,0,1,0), {\tt A} = (1,0,$-$1,0), {\tt B} = (1,0,i,0), {\tt 5} = (1,1,0,0), {\tt j} = (1,1,1,1), {\tt T} = (1,1,1,$-$1),\linebreak  {\tt ;} = (1,1,1,i), {\tt U} = (1,1,$-$1,1), {\tt k} = (1,1,$-$1,$-$1), {\tt "} = (1,1,$-$1,i), {\tt \$} = (1,1,i,1), {\tt  = } = (1,1,i,$-$1),\linebreak  {\tt l} = (1,1,i,i), {\tt 6} = (1,$-$1,0,0), {\tt f} = (1,$-$1,1,1), {\tt R} = (1,$-$1,1,$-$1), {\tt *} = (1,$-$1,1,i), {\tt S} = (1,$-$1,$-$1,1), {\tt g} = ($-$1,1,1,1), {\tt x} = (1,$-$1,$-$1,i), {\tt z} = (1,$-$1,i,1), {\tt /} = (1,$-$1,i,$-$1), {\tt h} = (1,$-$1,i,i), {\tt 7} = (1,i,0,0),\linebreak  {\tt r} = (1,i,1,1), {\tt b} = (1,i,1,$-$1), {\tt ]} = (1,i,1,i), {\tt c} = (1,i,$-$1,1), {\tt s} = (1,i,$-$1,$-$1), {\tt (} = (1,i,$-$1,i), {\tt \_} = (1,i,i,1),\linebreak  {\tt [} = (1,i,i,$-$1), {\tt t} = (1,i,i,i), {\tt K} = (i,0,0,1), {\tt C} = (i,0,1,0), {\tt $-$} = ($-$1,1,1,i), {\tt y} = ($-$1,1,$-$1,i),\linebreak  {\tt !} = ($-$1,1,i,1), {\tt :} = ($-$1,1,i,$-$1), {\tt i} = ($-$1,1,i,i), {\tt $<$} = ($-$1,$-$1,1,i), {\tt \#} = (i,i,i,1), {\tt \%} = ($-$1,$-$1,i,1),\linebreak  {\tt $>$} = (i,i,1,i), {\tt m} = (i,i,1,1), {\tt 8} = (i,1,0,0), {\tt n} = ($-$1,i,1,1), {\tt X} = ($-$1,i,1,$-$1), \ {\tt ?} = ($-$1,i,1,i),\linebreak  {\tt Y} = ($-$1,i,$-$1,1), {\tt o} = (i,1,i,i), {\tt \^{}} = (i,1,i,1), {\tt \&} = ($-$1,i,i,1), {\tt `} = (i,1,1,i), {\tt p} = (i,1,1,1), {\tt Z} = (i,1,1,$-$1),\linebreak  {\tt a} = (i,1,$-$1,1), {\tt q} = (i,1,$-$1,$-$1), {\tt '} = (i,1,$-$1,i), {\tt @} = (i,1,i,$-$1), {\tt u} = (i,$-$1,1,1),\linebreak  {\tt d} = (i,$-$1,1,$-$1),\hfill {\tt $\backslash$} = (i,$-$1,1,i),\hfill {\tt e} = (i,$-$1,$-$1,1),\hfill {\tt )} = (i,$-$1,i,1), \hfill{\tt V} = (i,i,1,$-$1), \hfill{\tt W} = (i,i,$-$1,1).

\bigskip
Distribution of critical KS MMPHs obtained (in an automated way
on a supercomputer with the help of VECFIND, STATES01, MMPSHUFFLE,
SHORTD, MMPTAG, and~Linux shell and Emacs macros) from
the {\bf 86-152} master, itself generated from $\{0,\pm 1,i\}$:

\medskip
{\parindent=0pt
  {\bf 9}  (number of hyperedges, $l$) 
  18 (number of vertices, $k$)
  (1) (number of MMPHs)

{\bf 11} 20 (2), 21 (2), 22 (4)

{\bf 13} 22 (2), 23 (6), 24 (33), 25 (40), 26 (35)

{\bf 15} 24 (1), 25 (3), 26 (52), 27 (208), 28 (573), 29 (542), 30 (265)

{\bf 16} 29 (1), 30 (7)

{\bf 17} 28 (4), 29 (103), 30 (860), 31 (2832), 32 (5011), 33 (3876), 34 (1325), 35 (1)

{\bf 18} 31 (7), 32 (11), 33 (25), 34 (15), 35 (19), 36 (5), 37 (2)

{\bf 19} 31 (2), 32 (112), 33 (724), 34 (2880), 35 (6701), 36 (9045), 37 (6139), 38 (1851), 39 (1)

{\bf 20} 34 (7), 35 (27), 36 (69), 37 (49), 38 (49), 39 (20), 40 (10), 41 (6)

{\bf 21} 34 (4), 35 (30), 36 (304), 37 (1318), 38 (3428), 39 (5807), 40 (6241), 41 (357),  42 (844), 43 (1)

{\bf 22} 37 (10), 38 (33), 39 (70), 40 (86), 41 (65), 42 (43), 43 (16), 44 (4), 45 (2)

{\bf 23} 38 (26), 39 (163), 40 (693), 41 (1617), 42 (3098), 43 (3749), 44 (3098), 45 (130), 46 (277), 47 (1)

{\bf 24} 39 (1), 40 (13), 41 (30), 42 (73), 43 (81), 44 (88), 45 (46), 46 (29), 47 (8), 48 (2), 49 (1)

{\bf 25} 40 (4), 41 (16), 42 (122), 43 (406), 44 (1005), 45 (1597), 46 (2127), 47 (175), 48 (1109), 49 (387), 50 (70), 51 (2)

{\bf 26} 41 (1), 42 (5), 43 (27), 44 (55), 45 (104), 46 (97), 47 (77), 48 (51), 49 (23), 50 (7), 52 (3)

{\bf 27} 42 (1), 43 (4), 44 (23), 45 (102), 46 (271), 47 (572), 48 (883), 49 (111), 50 (937), 51 (643), 52 (317), 53 (85), 54 (12)

{\bf 28} 44 (1), 45 (20), 46 (47), 47 (77), 48 (104), 49 (114), 50 (92), 51 (47), 52 (24), 53 (9), 54 (3), 55 (2)\linebreak
{\bf 29} 45 (1), 46 (9), 47 (40), 48 (118), 49 (230), 50 (349), 51 (506), 52 (53), 53 (462), 54 (314), 55 (149), 56 (83), 57 (18), 58 (2)

{\bf 30} 46  (1), 47  (10), 48  (33), 49  (79), 50  (139), 51  (157), 52  (141), 53  (97), 54  (44), 55  (20), 56  (7), 57  (1), 58  (3)

{\bf 31} 48 (2), 49 (35), 50 (86), 51 (157), 52 (250), 53 (272), 54 (297), 55 (26), 56 (206), 57 (148), 58 (69), 59 (35), 60 (13), 61 (1)

{\bf 32} 49  (4), 50  (22), 51  (75), 52  (149), 53  (193), 54  (196), 55  (126), 56  (69), 57  (37), 58  (17), 59  (12)

{\bf 33} 50 (5), 51 (24), 52 (67), 53 (146), 54 (226), 55 (219), 56 (196), 57 (15), 58 (126), 59 (85), 60 (59), 61 (25), 62 (10), 63 (1), 64 (2)

{\bf 34} 51  (4), 52  (18), 53  (54), 54  (130), 55  (220), 56  (200), 57  (188), 58  (12), 59  (64), 60  (31), 61  (9), 62  (2), 63  (1)

{\bf 35} 53 (8), 54 (41), 55 (118), 56 (184), 57 (240), 58 (219), 59 (158), 60 (11) 61 (67), 62 (33), 63 (18), 64 (7), 65 (1), 66 (4), 67 (1)

{\bf 36} 54 (5), 55 (26), 56 (91), 57 (194), 58 (232), 59 (226), 60 (166), 61 (89), 62 (39), 63 (23), 64 (4)

{\bf 37} 55 (1), 56 (23), 57 (71), 58 (155), 59 (226), 60 (248), 61 (185), 62 (10), 63 (58), 64 (22), 65 (15), 66 (7), 67 (3), 68 (1)

{\bf 38} 56(2), 57(14), 58  $\!$(48), 59  $\!$(143), 60  $\!$(194), 61  $\!$(221), 62 $\!$(192), 63  $\!$(10), 64 $\!$(62), 65 $\!$(22), 66 $\!$(5), 67 $\!$(1)

{\bf 39} 57 (3), 58 (10), 59 (35), 60 (84), 61 (174), 62 (210), 63 (156), 64 (13), 65 (60), 66 (29), 67 (8), 68 (5), 69 (2), 70 (2)

{\bf 40} 58 (2), 59 (2), 60 (18), 61 (55), 62 (140), 63 (167), 64 (139), 65 (10), 66 (64), 67 (20), 68 (6), 69 (11), 70 (1)

{\bf 41} 60 (2), 61 (7), 62 (34), 63 (89), 64 (125), 65 (116), 66 (101), 67 (67), 68 (34), 69 (13), 70 (3), 73 (1)\linebreak
{\bf 42} 62 (3), 63 (19), 64 (53), 65 (99), 66 (107), 67 (91), 68 (60), 69 (29), 70 (9), 72 (1)

{\bf 43} 63 (1), 64 (13), 65 (27), 66 (55), 67 (86), 68 (82), 69 (43), 70 (25), 71 (7), 72 (3), 73 (3), 75 (1)

{\bf 44} 64 (1), 65 (4), 66 (9), 67 (42), 68 (46), 69 (51), 70 (38), 71 (23), 72 (10), 73 (6), 75 (1)

{\bf 45} 66 (1), 67 (4), 68 (14), 69 (18), 70 (41), 71 (29), 72 (15), 73 (7), 74 (2)

{\bf 46} 66 (1), 68 (3), 69 (4), 70 (10), 71 (24), 72 (20), 73 (8), 74 (6), 76 (2)

{\bf 47} 69 (1), 70 (6), 71 (10), 72 (13), 73 (6), 74 (7), 75 (2), 76 (2)

{\bf 48} 71 (2), 72 (5), 73 (11), 74 (5), 75 (2), 76 (2), 77 (1), 78 (1), 79 (1)

{\bf 49} 74 (7), 75 (3), 76 (3)

{\bf 50} 73 (1), 76 (1), 77 (1), 81 (1)

{\bf 51} 75 (1)

{\bf 52} 76 (1)

111137 non-isomorphic critical KS MMPHs.}

\subsection{\label{app:3} 5-dim~MMPHs}

{\bf 29-16} {\tt 34125,56798,89CHG,GHF43,ABC95,DEF45,IJH94,KLG94,MNLB9,ONJA9,PMK79,\linebreak POI69,QRIE4,SRKD4,TQJ24,TSL14} {\tt 1} = (1,1,$-$1,1,0), {\tt 2} = (1,$-$1,1,1,0), {\tt 3} = (1,0,0,$-$1,0),\linebreak  {\tt 4} = (0,0,0,0,1), {\tt 5} = (0,1,1,0,0), {\tt 6} = (1,1,$-$1,0,1), {\tt 7} = (1,$-$1,1,0,1), {\tt 8} = (1,0,0,0,$-$1), {\tt 9} = (0,0,0,1,0), {\tt A} = (1,$-$1,1,0,$-$1), {\tt B} = (1,1,$-$1,0,$-$1), {\tt C} = (1,0,0,0,1), {\tt D} = (1,$-$1,1,$-$1,0), {\tt E} = (1,1,$-$1,$-$1,0),\linebreak  {\tt F} = (1,0,0,1,0), {\tt G} = (0,0,1,0,0), {\tt H} = (0,1,0,0,0), {\tt I} = (1,0,1,0,0), {\tt J} = (1,0,$-$1,0,0), {\tt K} = (1,1,0,0,0),\linebreak  {\tt L} = (1,$-$1,0,0,0), {\tt M} = (0,0,1,0,$-$1), {\tt N} = (1,1,1,0,1), {\tt O} = (0,1,0,0,$-$1), {\tt P} = ($-$1,1,1,0,1), {\tt Q} = (0,1,0,1,0), {\tt R} = (1,$-$1,$-$1,1,0), {\tt S} = (0,0,1,1,0), {\tt T} = (1,1,1,$-$1,0).

\bigskip
{\bf 11-7} \qquad  {\tt 789AB,B65,513,32897,46A,134,245.}

\bigskip
{\bf 21-7} (filled 11-7) {\tt 789AB,5CD6B,4EF6A,23789,1GH34,1IJ35,2KL45.} 1 = (1,1,0,0,0),\linebreak  {\tt 2} = (0,1,0,0,0), {\tt 3} = (0,0,1,0,0), {\tt 4} = (0,0,0,1,1), {\tt 5} = (0,0,0,1,$-$1), {\tt 6} = (1,0,0,0,0), {\tt 7} = (1,0,0,$-$1,0),\linebreak  {\tt 8} = (0,0,0,0,1), {\tt 9} = (1,0,0,1,0), {\tt A} = (0,1,$-$1,0,0), {\tt B} = (0,1,1,0,0), {\tt C} = (0,1,$-$1,1,1), {\tt D} = (0,$-$1,1,1,1), {\tt E} = (0,1,1,1,$-$1), {\tt F} = (0,1,1,$-$1,1), {\tt G} = (1,$-$1,0,1,$-$1), {\tt H} = (1,$-$1,0,$-$1,1), {\tt I} = (1,$-$1,0,1,1),\linebreak  {\tt J} = ($-$1,1,0,1,1), \ {\tt K} = (1,0,1,0,0), \ {\tt L} = (1,0,$-$1,0,0).

\subsection{\label{app:4} 6-dim~MMPHs}

{\bf 33-11} {\tt 123456,6789AB, BCD3EF, G5FMNO,8HIXWV, IAJD2K,KE4G7L,LH9JC1,MNOPQR,\linebreak PQRSTU,STUVWX.}
{\tt 1}=(0,0,1,1,$\omega$,$\omega$), {\tt 2} = (0,1,0,$\omega$,1,$\omega$), {\tt 3} = (0,1,$\omega$,0,$\omega$,1), {\tt 4} = (0,$\omega$,1,$\omega$,0,1),\linebreak  {\tt 5} = (1,0,0,0,0,0), {\tt 6} = (0,$\omega$,$\omega$,1,1,0), {\tt 7} = (1,1,$\omega$,$\omega$,0,0), {\tt 8} = (0,0,0,0,0,1), {\tt 9} = ($\omega$,0,1,$\omega$,1,0),\linebreak  {\tt A} = ($\omega$,1,0,1,$\omega$,0), {\tt B} = (1,$\omega$,1,0,$\omega$,0), {\tt C} = (1,0,$\omega$,0,1,$\omega$), {\tt D} = ($\omega$,$\omega$,0,0,1,1), {\tt E} = ($\omega$,1,1,0,0,$\omega$),\linebreak  {\tt F} = (0,0,0,1,0,0), {\tt G} = (0,0,0,0,1,0), {\tt H} = (0,1,0,0,0,0), {\tt I} = (0,0,1,0,0,0), {\tt J} = (1,0,0,$\omega$,$\omega$,1),\linebreak  {\tt K} = (1,$\omega$,0,1,0,$\omega$), {\tt L} = ($\omega$,0,$\omega$,1,0,1), {\tt M} = (0,1,$\omega$,0,0,$\omega$), {\tt N} = (0,$\omega$,1,0,0,$\omega$), {\tt O} = (0,$\omega$,$\omega$,0,0,1), {\tt P} = (1,0,0,$\omega$,$\omega$,0), {\tt Q} = ($\omega$,0,0,1,$\omega$,0), {\tt R} = ($\omega$,0,0,$\omega$,1,0), {\tt S} = (0,1,1,0,0,$\omega$), {\tt T} = (0,1,$\omega$,0,0,1),\linebreak  {\tt U} = (0,$\omega$,1,0,0,1), \ {\tt V} = (1,0,0,1,$\omega$,0), \ {\tt W} = (1,0,0,$\omega$,1,0), \ {\tt X} = ($\omega$,0,0,1,1,0)

\bigskip
{\bf 24-11}\hfil ($\overline{\rm subhypergraph}$ of 32-11)\hfil   {\tt 123456,2568BM,MKAGCH,CGHJ9I,I7FNDO,NDOEL1,3LDJCK,\linebreak 456789,ABCDEF,A5,A6.}

\bigskip
{\bf 32-11} {\tt 123456,2568BM,MKAGCH,CGHJ9I,I7FNDO,NDOEL1,3LDJCK,456789,ABCDEF, \linebreak AWVUT5,ASRQP6.} {\tt 1} = (1,1,$-$1,0,0,1), {\tt 2} = (1,0,0,0,0,$-$1), {\tt 3} = (1,$-$1,1,0,0,1), {\tt 4} = (0,1,1,0,0,0),\linebreak  {\tt 5} = (0,0,0,0,1,0), {\tt 6} = (0,0,0,1,0,0), {\tt 7} = (1,$-$1,1,0,0,$-$1), {\tt 8} = (1,0,0,0,0,1), {\tt 9} = (1,1,$-$1,0,0,$-$1), {\tt A} = (0,0,0,0,0,1), {\tt B} = (0,0,1,0,0,0), {\tt C} = (0,0,0,1,$-$1,0), {\tt D} = (0,0,0,1,1,0), {\tt E} = (1,$-$1,0,0,0,0),\linebreak  {\tt F} = (1,1,0,0,0,0), {\tt G} = (1,0,1,1,1,0), {\tt H} = (1,0,1,$-$1,$-$1,0), {\tt I} = (1,$-$1,$-$1,0,0,1), {\tt J} = (0,1,0,0,0,1),\linebreak  {\tt K} = (1,0,$-$1,0,0,0), {\tt L} = (1,1,1,0,0,$-$1), {\tt M} = (0,1,0,0,0,0), {\tt N} = (0,0,1,1,$-$1,1), {\tt O} = (0,0,1,$-$1,1,1),\linebreak  {\tt P} = (0,1,0,0,1,0), {\tt Q} = (1,0,1,0,0,0), {\tt R} = (1,1,$-$1,0,$-$1,0), {\tt S} = (1,$-$1,$-$1,0,1,0), {\tt T} = (0,1,1,1,0,0),\linebreak  {\tt U} = (1,0,1,$-$1,0,0), {\tt V} = (1,1,$-$1,0,0,0), {\tt W} = (1,$-$1,0,1,0,0).

\subsection{\label{app:5} 8-dim~MMPHs}

{\bf 34-9} {\tt 12345678,9ABC5678,DEFG3478,HIJKLMFG,NOPQRMEC,STUQRLDB,VWUPJK9A, \linebreak XYWTOI28,XYVSNH17.} 1 = (0,0,1,1,1,1,0,0), {\tt 2} = (0,0,1,$-$1,1,$-$1,0,0), {\tt 3} = (0,0,0,1,0,$-$1,0,0),\linebreak  {\tt 4} = (0,0,1,0,$-$1,0,0,0), {\tt 5} = (0,1,0,0,0,0,0,0), {\tt 6} = (1,0,0,0,0,0,0,0), {\tt 7} = (0,0,0,0,0,0,0,1),\linebreak  {\tt 8} = (0,0,0,0,0,0,1,0), {\tt 9} = (0,0,0,1,0,0,0,0), {\tt A} = (0,0,1,0,0,0,0,0), {\tt B} = (0,0,0,0,0,1,0,0),\linebreak  {\tt C} = (0,0,0,0,1,0,0,0), {\tt D} = (1,$-$1,1,0,1,0,0,0), {\tt E} = (1,1,0,1,0,1,0,0), {\tt F} = (1,1,0,$-$1,0,$-$1,0,0),\linebreak  {\tt G} = ($-$1,1,1,0,1,0,0,0), {\tt H} = (0,1,$-$1,1,0,0,1,0), {\tt I} = (1,0,1,1,0,0,0,$-$1), {\tt J} = (1,0,0,0,1,1,0,1),\linebreak  {\tt K} = (0,1,0,0,$-$1,1,$-$1,0), {\tt L} = (0,0,1,0,$-$1,0,1,1), {\tt M} = (0,0,0,1,0,$-$1,$-$1,1), {\tt N} = (1,0,1,0,0,$-$1,1,0),\linebreak  {\tt O} = (0,$-$1,1,0,0,1,0,1), {\tt P} = ($-$1,1,0,0,0,0,1,1), {\tt Q} = (1,0,$-$1,$-$1,0,0,0,1), {\tt R} = (0,1,1,$-$1,0,0,$-$1,0),\linebreak  {\tt S} = (1,0,0,1,$-$1,0,$-$1,0), {\tt T} = (0,1,0,1,1,0,0,1), {\tt U} = (1,1,0,0,0,0,1,$-$1), {\tt V} = (0,1,0,0,1,$-$1,$-$1,0),\linebreak  {\tt W} = (1,0,0,0,$-$1,$-$1,0,1), {\tt X} = (1,1,0,$-$1,0,1,0,0), {\tt Y} = (1,$-$1,$-$1,0,1,0,0,0).

\bibliographystyle{quantum}

\end{document}